\documentclass[
notitlepage,
floatfix,
aps,
prl,
reprint,
twocolumn,
superscriptaddress,
nofootinbib]{revtex4-2}
\PassOptionsToPackage{usenames,dvipsnames}{color}
\usepackage{amsmath,amsfonts,amssymb,amsthm,graphicx,bbm,enumerate,times}
\usepackage[normalem]{ulem}
\usepackage{mathtools}
\usepackage{import}
\usepackage{hhline}
\usepackage{hyperref}
\usepackage{enumitem}
\usepackage{svg}
\usepackage{tikz}
\usepackage{ulem}
\usetikzlibrary{shapes.geometric}
\usetikzlibrary{shapes.symbols}
\usepackage{wasysym}
\usepackage{makecell}
\usepackage{lipsum}
\usepackage{nicematrix}
\NiceMatrixOptions{cell-space-limits = 1pt}

\newcommand{\Sl}{\mathrm{Sl}}

\newcommand{\tp}{\intercal}

\newtheorem{theorem}{Theorem}
\newtheorem{proposition}{Proposition}

\newtheorem{lemma}{Lemma}

\newtheorem{definition}{Definition}

\definecolor{DarkGreen}{rgb}{0.5, 0.5, 0}
\definecolor{teal}{rgb}{0.0, 0.5, 0.5}
\newcommand{\mc}[1]{\mathcal{#1}}

\newcommand{\mb}[1]{\mathbb{#1}}
\newcommand{\mbf}[1]{\mathbf{#1}}

\newcommand{\tr}{\mathrm{Tr}} %old
\newcommand{\Tr}{\mathrm{Tr}} %new
\newcommand{\supp}{\mathrm{supp}}

\newcommand{\iden}{\mb{I}}

\newcommand{\1}{\mathbf{1}}
\newcommand{\toGCO}{\stackrel{\mathrm{GCO}}{\longrightarrow}}
\newcommand{\Inm}{m_\mathrm{in}}
\newcommand{\Outm}{m_{\mathrm{out}}}

\newcommand{\mU}{\mc{U}}
\newcommand{\mE}{\mc{E}}

\newcommand{\mmZ}{\mbf{Z}}

\newcommand{\ket}[1]{\left.\left|{#1}\right.\right\rangle}

\newcommand{\bra}[1]{\left.\left\langle{#1}\right.\right|}

\newcommand{\vecr}{\vec{r}}
\newcommand{\vecalpha}{\vec{\alpha}}
\newcommand{\tX}{X}
\newcommand{\tY}{Y}
\newcommand{\td}{\vec{d}}
\newcommand{\tS}{S}
\newcommand{\tR}{R}

\makeatletter
\def\maketitle{
	\@author@finish
	\title@column\titleblock@produce
	\suppressfloats[t]}
\makeatother

\newcommand{\ntu}{Nanyang Quantum Hub, School of Physical and Mathematical Sciences, Nanyang Technological University, 637371, Singapore}
\newcommand{\ulm}{Institute of Theoretical Physics \& IQST, Ulm University, Albert-Einstein-Allee 11 89081, Ulm, Germany}
\newcommand{\shandong}{School of Information Science and Engineering, Shandong University, Qingdao 266237, China}
\newcommand{\cqt}{Centre for Quantum Technologies, Nanyang Technological University, 50 Nanyang Avenue, 639798 Singapore}

\begin{document}

	\title{Gaussian time-translation covariant operations: structure, implementation, and thermodynamics}

	\author{Xueyuan Hu}
    \email{xyhu@sdu.edu.cn}
	\affiliation{\shandong}
	\author{Lea Lautenbacher}
    \email{lea.lautenbacher@uni-ulm.de}
	\affiliation{\ulm}
	\author{Giovanni Spaventa}
    \email{giovanni.spaventa@uni-ulm.de}
	\affiliation{\ulm}
	\author{Martin B. Plenio}
    \email{martin.plenio@uni-ulm.de}
	\affiliation{\ulm}
	\author{Nelly H.Y. Ng}
    \email{nelly.hy.ng@gmail.com}
	\affiliation{\ntu}
	\affiliation{\cqt}
	\author{Jeongrak Son}
	\email{jeongrak.son@gmail.com}
	\affiliation{\ntu}
    \affiliation{Dahlem Center for Complex Quantum Systems, Freie Universität Berlin, 14195 Berlin, Germany}
	\begin{abstract}
    Time-translation symmetry strongly constrains physical dynamics, yet systematic characterization for continuous-variable systems lags behind its discrete-variable counterpart. We close this gap by providing a rigorous classification of Gaussian quantum operations that are covariant under time translations, termed Gaussian covariant operations. We show that several key results known for discrete-variable covariant operations break down in the Gaussian optical setting: discrepancies arise in physical and thermodynamic implementation, in the extensivity of asymmetry, and in catalytic advantages. Our results provide comprehensive mathematical and operational toolkits for Gaussian covariant operations, including a peculiar pair of asymmetry measures that are completely non‑extensive. Our findings also reveal surprising consequences of the interplay among symmetry, Gaussianity, and thermodynamic constraints, suggesting that real-world scenarios with multiple constraints have a rich structure not accessible from examining individual constraints separately. 
	\end{abstract}
	\maketitle

Time flows whether we have a clock or not.
Without a clock, however, we cannot keep track of temporal information, which severely limits our ability to control the state.
Suppose that a clock-less agent applies a quantum channel to a system that is prepared in a known initial state, but the agent does not know how much time has passed since preparation.
Then any information the agent has about the system after the operation must be averaged over all possible timings of the channel.
Mathematically, this corresponds to time-averaged versions of channels, which form a class of $\mathrm{U}(1)$-covariant operations, hereafter referred to as phase-covariant operations~\cite{Bartlett07}. 

To access the full information, timekeeping devices such as clocks are needed. 
Then the operations before the time-averaging, which are not covariant in general, can be considered. 
A common way to model such processes is to treat the timekeeping device as an additional quantum system and describe the overall evolution on the enlarged system by a globally covariant operation. Notably, timekeeping devices generally degrade with use unless additional resources are supplied~\cite{Erker2017Clock, Woods2019_Clocks}.
Thus, to study the cost of these devices, a good understanding of covariant operations is crucial. 

Quantum resource theories~\cite{Gour25} are developed precisely for such careful accounting of resources. 
The properties of covariant operations and how they utilize timekeeping devices are extensively studied in the context of the resource theory of asymmetry~\cite{Bartlett07, Gour08, MarvianThesis}.
This approach has proved useful for understanding quantum metrology~\cite{Streltsov2017_Coherence, Marvian2014_Noether}, quantum clocks and reference frames~\cite{Bartlett07}, and the cost of realizing quantum measurements~\cite{Tajima22}.
Most resource-theoretic studies to date have focused on finite-dimensional systems, whilst time-translation symmetry for continuous-variable systems, especially Gaussian optical ones, has been studied separately under the name phase-insensitive operators.
In particular, seminal results, such as studies on linear amplifiers~\cite{Haus1962, Caves1982, Clerk2010Amp}, the resolution of the minimum output entropy conjecture for Gaussian bosonic channels~\cite{Giovannetti04, Giovannetti13, DePalma14, DePalma15}, extensive studies of the classical and quantum capacities of such channels~\cite{Holevo01, Giovannetti04Cap}, and analyses of boson sampling~\cite{Aaronson10, Hamilton17}, have relied heavily on time-translation symmetry.
Yet, Gaussian operations that are time-translation symmetric---which we call Gaussian covariant operations in the rest of this letter---have rarely been studied in their full generality. 

In this work, we unify and generalize existing studies on Gaussian phase-insensitive operations into a framework that encompasses input and output systems with arbitrary numbers of modes.
In particular, we establish both mathematical and physical characterizations of Gaussian covariant operations, namely i) a necessary and sufficient condition for them to be freely dilatable and ii) a monotone function that does not increase under Gaussian covariant operations and exhibits several peculiar properties.
Each of these results showcases surprising deviations from general time-translation symmetric operations and from Gaussian operations when both conditions are imposed.

Lastly, we apply our characterization results to thermodynamics, where time-translation symmetry plays a central role~\cite{Lostaglio15, Korzekwa2016, Smirne16, Kwon2018, Shiraishi2025GPC}. 
Since time-translation-covariant unitaries coincide with energy-preserving unitaries, demanding the first law (total energy preservation) amounts to requiring operations to be time-translation covariant unless free coherence is supplied via auxiliary systems. 
Moreover, a version of the second law can be encoded as Gibbs‑state preservation (non-increasing non‑equilibrium free energies).
Enhanced thermal operations~\cite{EnTO2015} are defined as the set of all operations satisfying these two conditions.
However, not all enhanced thermal operations are freely implementable: there exists a gap~\cite{Ding21} between them and the set of operations realizable with a Gibbs-state ancilla and an energy-preserving unitary, known as thermal operations~\cite{Janzing00_TO}.
Understanding the exact nature of this gap (beyond its origin in asymmetry manipulation) remains a long-standing open problem~\cite{Lostaglio15Symmetry, Gour2018EnTO, Lie2025TO}; progress here will elucidate the genuinely quantum aspects of thermodynamic operations.
In Gaussian optics, time-translation-covariant unitaries correspond exactly to passive elements~\cite{SerafiniBook, Weedbrook12}, and the Gaussian analogue of thermal operations has been introduced~\cite{Serafini20, Narasimhachar2021GTO, Yadin22}.
We show that, in this setting, the gap closes: Gaussian thermal operations and Gaussian enhanced thermal operations coincide, resolving the question for a large and experimentally relevant class.

\emph{Notations}---%
We consider continuous-variable systems consisting of $m$ bosonic modes with Hamiltonian $H=\sum_{j=1}^m\omega_j\hat{a}_j^\dagger \hat{a}_j$, where $\hat a_j$ and $\hat a_j^\dagger$ are the annihilation/creation operators and $\omega_j> 0$ is the frequency of the $j$th mode. 
A Gaussian state of this system is fully characterized by its first and second moments defined as $r_{j} = r_{m+j}^{*} = \langle \hat{a}_{j} \rangle$ for $j = 1,\cdots,m$ and $M_{j,k} = M_{k+m,j+m}= \frac{1}{2}\langle \{\hat{a}_{j},\hat{a}_{k}^{\dagger}\} \rangle - \langle \hat{a}_{j}\rangle \langle \hat{a}_{k}^{\dagger}\rangle $, $M_{j,k+m} = M_{j+m,k}^{*} = \frac{1}{2}\langle \{\hat{a}_{j},\hat{a}_{k}\} \rangle - \langle \hat{a}_{j}\rangle \langle \hat{a}_{k}\rangle $ for $j,k = 1,\cdots,m$, respectively.
Accordingly, the first moment vector $\vecr$ and the second moment matrix $M$ can be written as 
\begin{align}\label{eq:cvm}
    \vecr=\begin{pmatrix}
        \vecalpha^* \\ \vecalpha
    \end{pmatrix},\ M = \begin{pmatrix}
        \mu^* & \chi \\ \chi^* & \mu
    \end{pmatrix},
\end{align}
where $\vecalpha$ is a $m$-dimensional vector and $\mu$ and $\chi$ are $m\times m$-dimensional matrices that are Hermitian and symmetric, respectively. 
Hence, a Gaussian state is fully characterized by $(\vec\alpha,\mu,\chi)$.
Moreover, the second moment must obey the uncertainty principle, which requires $M+\mathbf{Z}\geq \mathbf{0}$. Here $\mathbf{Z}\coloneqq \frac{1}{2}\begin{pmatrix}
	\mathbb{I} & \mathbf{0}\\ \mathbf{0} & -\mathbb{I}
\end{pmatrix}$ and $\mathbb{I}$ is the $m\times m$ identity matrix.
Note that we use the first and the second moments defined with $\hat{a}_j$, $\hat{a}_j^\dagger$ operators instead of the more popular choice of $\hat{x}_{j} = \frac{1}{\sqrt{2}}(\hat{a}_{j}^{\dagger}+\hat{a}_{j})$, $\hat{p}_{j} = \frac{i}{\sqrt{2}}(\hat{a}_{j}^{\dagger} - \hat{a}_{j})$; these two representations are equivalent up to a simple transformation and we choose the former for notational simplicity in most of our proofs. 

Similarly, Gaussian operations are specified by the tuple $(\tX, \tY,\td)$ acting on the first and second moments as $\vecr \rightarrow \tX \vecr + \td$ and $M\rightarrow \tX M \tX^\dagger + \tY$, with the constraint $\tY \geq \tX \mathbf{Z} \tX^\dagger - \mathbf{Z}$. 
We also allow scenarios where the input and the output modes are different. 
These operations emerge when e.g. input system $S_\mathrm{in}$ interact with ancilla system $S_\mathrm{out}R$ followed by the partial trace of $S_\mathrm{in}R$.
In such cases, $\tX$ may not be square matrices, and the constraint becomes $\tY \geq \tX \mathbf{Z}_\mathrm{in} \tX^\dagger - \mathbf{Z}_\mathrm{out}$. 
For Gaussian unitary operations, the tuple simplifies to $(\tS, 0,\td)$, where $\td$ describes the displacement operation, and $\tS$ (satisfying $\tS\mmZ \tS^\dagger=\mmZ$) is the total effect from beam-splitters, phase-shifters and squeezing. 

\emph{Mathematical Characterization}---%
In this work, we are mainly interested in operations that are covariant to time-translation under the free Hamiltonian.
A \emph{Gaussian covariant operation (GCO)} can be formally defined as a Gaussian channel $\mE_{\mathrm{GCO}}$ that commutes with time-translation: 
\begin{align}\label{eq:GCO_condition}
    \mE_{\mathrm{GCO}}\circ\mathcal{U}_t^\mathrm{in}=\mathcal{U}_t^\mathrm{out}\circ\mE_{\mathrm{GCO}}
\end{align}
for all $t\in\mathbb{R}$, where $\mathcal{U}_t^\mathrm{in/out}(\cdot)=e^{-iH_\mathrm{in/out}t}(\cdot)e^{iH_\mathrm{in/out}t}$, and $H_\mathrm{in/out}$ are the free Hamiltonians of the input and output systems of $\mE_{\mathrm{GCO}}$.

The commutation relation Eq.~\eqref{eq:GCO_condition} greatly simplifies the structure of GCO. 
Proposition~\ref{pro: indep evols} in Appendix~\ref{ap:proofs1} demonstrates that input modes $j,k$ of different frequencies $\omega_{j}\neq\omega_{k}$ evolve independently under GCO. 
It is therefore sufficient to study the case where all input and output modes have the same frequency $\omega$.
Hence, the time-translation operation $\mathcal{U}_{t}$ is equivalent to an overall phase shift by $\theta = \omega t$, corresponding to a tuple $(\tR(\theta),0,0)$, where $\tR(\theta) = \cos(\theta) \mathbb{I}-i\sin(\theta) \cdot 2\mathbf{Z}$.

The mathematical characterization of GCOs is derived in the literature~\cite{Heinosaari2010GaugeCov, Holevo2015GaugeCov, Frank2017GaugeCov} with the name Gaussian gauge-covariant operations.
In our notation, the characterization can be stated as follows.
Any $\Inm$-mode to $\Outm$-mode GCO is \emph{fully characterized} by a pair $(A,B)$ of $\Outm\times \Inm$ matrix $A$ and $\Outm\times \Outm$ matrix $B$, transforming first and second moments as: 
    \begin{align}
        \vecalpha &\rightarrow A\vecalpha,\label{eq:1st_moment}\\
        \mu & \rightarrow   A\mu A^\dagger + B,\label{eq:2nd_moment_mu}\\
        \chi &\rightarrow A^* \chi A^\dagger,\label{eq:2nd_moment_chi}
    \end{align}
    with the constraint
    \begin{align}
    	B\geq \pm\frac{1}{2}(  \iden-AA^\dagger).\label{eq:HUR_A_B}
	\end{align}
For completeness, we provide a self-contained proof in Lemma~\ref{le:gene_GCO_app}, Appendix~\ref{ap:proofs1}.
This characterization greatly simplifies the analysis of GCO by reducing the tuple $(X,Y,\td)$ to the much smaller pair $(A,B)$ and showing that second moments $\mu$ and $\chi$ evolve independently.
Furthermore, if a GCO $(A,B)$ is unitary, corresponding to passive elements in optics or energy-preserving unitaries in thermodynamics, the reversibility of unitary operations demands $B=0$, and in turn, $AA^\dagger=\iden$ from Eq.~\eqref{eq:HUR_A_B}. 
Therefore, a unitary GCO is fully characterized by a single $m\times m$ unitary matrix $A=V$, which acts on the first and second moments as $\vecalpha\rightarrow V\vecalpha$, $\mu\rightarrow V\mu V^\dagger$, and $\chi \rightarrow V^* \chi V^\dagger$. 

\emph{Physical and operational characterization}---
Beyond the mathematical characterization, we provide operational ones.
Typically, a desired quantum channel is implemented by its dilation, i.e., by preparing an auxiliary state and applying a unitary.
For this dilation to be free, the auxiliary and the unitary both need to be free.
We define a class of GCOs with such free dilations (also known as physically implementable operations in the literature~\cite{Gour25}). 
\begin{definition}[Freely dilatable GCO]\label{def: free dilation}
	A GCO $\mE$ on system $S$ has a \emph{free dilation}, if it can be written as 
	\begin{align}\label{eq: PI GCO def}
		\mE(\cdot) = \tr_R\left[\mathbf{U}^{\mathrm{PC}}(\cdot\otimes\rho_{R})\mathbf{U}^{\mathrm{PC}\dagger}\right],
	\end{align}
where $\rho_{R}$ is a Gaussian symmetric state satisfying $[\rho_{R},H_{R}] = 0$
for the Hamiltonian $H_{R}$, and $\mathbf{U}^{\mathrm{PC}}$ is a Gaussian covariant unitary on $SR$, i.e., $[\mathbf{U}^{\mathrm{PC}},H_S+H_R]=0$. 
\end{definition}

Here symmetric states are those invariant under time-translation, or equivalently, under phase-shifts.
In our $(\vec{\alpha},\mu,\chi)$ notation, this corresponds to $\vec{\alpha} = 0$ and $\chi = 0$.
In the following lemma (proved in Lemma~\ref{le:free_dilation_app}, Appendix~\ref{ap:proofs1_2}), we provide a complete characterization of the GCOs that admit a free dilation.
\begin{lemma}\label{le:free_dilation}
A GCO $(A,B)$ is freely dilatable if and only if the following two conditions are satisfied.\\
(F1) $\iden-AA^\dagger\geq0$;\\
(F2) $\supp(B)= \supp(\iden-AA^\dagger)$.
\end{lemma}

Using Lemma~\ref{le:free_dilation}, we also identify an operationally relevant class of GCOs that do not admit a free dilation in Theorem~\ref{th:GCO_PI} (proved in Theorem~\ref{th:GCO_PI_app}, Appendix~\ref{ap:proofs1_2}).

\begin{theorem}
\label{th:GCO_PI}
    If a GCO mapping a system to itself has no fixed-point that is a valid quantum state, then it is not freely dilatable. 
    However, the converse is not true.
\end{theorem}
Theorem~\ref{th:GCO_PI} is a concrete instance of Lemma~\ref{le:free_dilation} and provides the physical intuition behind the impossibility of a free implementation.
We present this intuition here, in summary, and in Appendix~\ref{ap:proofs1_2}, in detail. 
If a GCO (from a system to itself) has no fixed-point quantum state, it must be capable of increasing the energy of the system’s hottest mode, irrespective of how large that energy already is.
Conversely, if a GCO is freely dilatable, it is implemented by a Gaussian covariant unitary $\mathbf{U}^{\mathrm{PC}}$ as in Eq.~\eqref{eq: PI GCO def}, and Gaussian covariant unitaries (i.e. beam splitters and phase shifters) cannot further excite the hottest mode.
Consequently, when the system’s hottest mode is initially hotter than any mode of the ancilla, its energy cannot increase under a freely dilatable GCO.  
This contradicts the required energy‑increase of a GCO without a fixed point.

Lemma~\ref{le:free_dilation} and Theorem~\ref{th:GCO_PI} reveal that not all GCOs admit a free dilation. 
This is not the case for general time-translation covariant operations, where any covariant map (including infinite-dimensional ones) can be dilated using a symmetric auxiliary and a covariant unitary~\cite{Keyl99, MarvianThesis}.
Nor is it the case for general Gaussian operations, which all admit a dilation with a Gaussian auxiliary and a Gaussian unitary~\cite{SerafiniBook}.
These discrepancies follow from the intuition above: covariant unitaries and Gaussian unitaries can further excite the system’s hottest mode, whereas unitaries that are Gaussian \emph{and} covariant cannot.

Lemma~\ref{le:free_dilation} and Theorem~\ref{th:GCO_PI} also sharply contrast with recent results that every GCO admits a Gaussian-covariant dilation~\cite{Koukoulekidis2025}. 
The main difference lies in the allowed auxiliaries.
In our framework frequencies are required to be positive; Ref.~\cite{Koukoulekidis2025} allows auxiliary modes with negative frequencies, making the energy spectrum unbounded from below. 
Although unphysical in the Schr\"{o}dinger picture, negative effective frequencies can appear in a rotating (modulation) frame~\cite{PhysRevA.31.3068,PhysRevX.2.031016}. 
In that picture two bosonic modes can have effective frequencies $\pm \omega$ and two-mode squeezing unitaries become covariant under the modulated free Hamiltonian.
The most salient examples are Gaussian amplifiers~\cite{Holevo2015GaugeCov}.
They increase the system's energy regardless of the initial state and thus have no fixed-point.
Theorem~\ref{th:GCO_PI} then implies that they do not admit a free dilation; however, within the framework of Ref.~\cite{Koukoulekidis2025}, they are freely dilatable. 
The main difference is that a two‑mode‑squeezing unitary---available when negative‑frequency modes are allowed---can further excite the hottest mode, unlike passive linear unitaries.
In this case, the conserved quantity associated with the time‑translation symmetry is no longer the original free Hamiltonian but the modulated one.
Consequently, implementations of covariant operations defined in the modulation picture generally require external energy source (e.g. a strong pump beam at high frequency~\cite{PhysRevA.31.3068}) and are therefore not thermodynamically free.
In contrast, $\mathbf{U}^{\mathrm{PC}}$ in Definition~\ref{def: free dilation} is both covariant and energy-preserving, making it more relevant for thermodynamic considerations.

We can attempt to measure the resource cost required for implementing non-freely dilatable (but resource non-generating) GCOs. 
The cost can be measured as the amount of resource that must be supplied in the auxiliary state to implement the desired operation; such a framework is standard in quantifying thermodynamic~\cite{Bennett82, Faist15, Faist18} and asymmetry cost~\cite{Tajima20, Tajima22, Tajima25GP}. 
In the case of GCOs, Lemma~\ref{le:free_dilation_app} (Appendix~\ref{ap:proofs1_2}) shows that any auxiliary state (non‑Gaussian and/or asymmetric) cannot implement a GCO without a free dilation using Gaussian covariant unitaries.
In other words, the cost of implementing such operations is infinite.

On the other hand, \emph{Gaussian thermal operations (GTO)}~\cite{Serafini20, Narasimhachar2021GTO} constitute an important class of freely dilatable GCOs. 
A GTO is defined as a freely dilatable GCO with the constraint that the ancilla $\rho_{R}$ is a Gibbs state at ambient temperature $\beta^{-1}$, i.e., $\rho_{R} = e^{- \beta H_R}/\tr[{e^{-\beta H_R}}]$.
It is the Gaussian analogue of thermal operations~\cite{Janzing00_TO} defined as channels implementable with an energy-preserving unitary and a Gibbs state ancilla.
Thermal operations exactly constitute the set of freely dilatable operations in thermodynamic theories, yet they are notoriously hard to characterize, as are the state transformations resulting from them.
This motivated the definition of a slightly larger set known as enhanced thermal operations~\cite{EnTO2015}.
Enhanced thermal operations are defined axiomatically as the set of covariant operations that preserve the Gibbs state at the ambient temperature, and are therefore easier to characterize~\cite{Gour2018EnTO}.
Unfortunately, the set of enhanced thermal operations is strictly larger than that of thermal operations

Unfortunately, the set of enhanced thermal operations is strictly larger than that of thermal operations at the level of state transformations, i.e. there exists a pair of states that an enhanced thermal operation can map, while no thermal operation can achieve even an approximate version of this transformation~\cite{Ding21}.
This separation is stronger than the one at the channel level, which only implies that some enhanced thermal operations are not thermal operations and does not rule out the possibility that every state transformation achievable by an enhanced thermal operation could be reproduced by a different thermal operation.

The \emph{Gaussian analogue of enhanced thermal operations (GEnTO)} can be defined similarly, as the subset of GCO that preserves a Gibbs state.
Surprisingly, we show that every GEnTO is freely dilatable using thermodynamically free and Gaussian resources.

\begin{theorem}\label{th:gento_gto}
A Gaussian enhanced thermal operation is always a Gaussian thermal operation.
\end{theorem}
A more technical statement would be the following. 
Let $\mE_{\mathrm{GEnTO}}$ be a GEnTO on system $S$.
Then there exist an auxiliary Gibbs state $\gamma_R=\frac{e^{-\beta H_R}}{\tr[e^{-\beta H_R}]}$ and a unitary GCO $\mathbf{U}^{\mathrm{PC}}$, such that $\mE_{\mathrm{GEnTO}}(\cdot)=\tr_R[\mathbf{U}^{\mathrm{PC}}(\cdot\otimes\gamma_R)\mathbf{U}^{\mathrm{PC}\dagger}]$.
See Theorem~\ref{th:gento_gto_app}, Appendix~\ref{ap: EnTO} for a proof. 
This result contrasts with the scenario without the Gaussian constraint, where a gap exists between thermal operations and enhanced thermal operations at the level of state transformations. 
Theorem~\ref{th:gento_gto} indicates that imposing the Gaussian-preserving constraint on both classes removes this gap and, moreover, establishes a stronger equivalence at the level of channels.

Consider the simplest case of a single-mode GCO $(A,B)$.
If it has a physical fixed-point $(\vec{\alpha},\mu,\chi)$, it also fixes a Gaussian symmetric state $(0,\mu,0)$, which is a single-mode Gibbs state at temperature $\beta(\mu)=\frac{1}{\omega}\ln\frac{\mu-1/2}{\mu+1/2}$.
Thus Theorem~\ref{th:gento_gto} implies that any single-mode GCO with a fixed-point $(\vec{\alpha},\mu,\chi)$ is a GTO implemented at temperature $\beta(\mu)$, and consequently, is freely dilatable.

Unlike the single-mode scenario, the existence of a fixed-point is no longer sufficient to guarantee physical implementability for multi-mode systems, as shown in Theorem~\ref{th:GCO_PI}. 
Nevertheless, Theorem~\ref{th:gento_gto} guarantees that if a Gibbs state is one of the fixed-points, a GCO is freely dilatable. 

\emph{Completely non-extensive monotones}---%
We have so far focused on the properties of GCO channels themselves; below we analyze second moments transformations $(\mu,\chi)\stackrel{ \mathrm{GCO}}{\longrightarrow}(\mu',\chi')$ induced by these channels.
For this purpose, we identify a pair of monotones.

\begin{definition}\label{def: Sl monotones}
	Let $(\mu,\chi)$ be the second moment of a quantum state. 
	$\Sl_{\pm}(\mu,\chi)$ are defined as 
	\begin{align}\label{eq: Sl monotone def}
		\Sl_{\pm}(\mu,\chi) \coloneqq \boldsymbol{\sigma}_{1}\left[ \left(\mu^{*}\pm\frac{\iden}{2}\right)^{-\frac{1}{2}}\chi\ \left(\mu\pm\frac{\iden}{2}\right)^{-\frac{1}{2}}\right],
	\end{align}
	where $(\cdot)^{-\frac{1}{2}}$ denotes the pseudoinverse of $(\cdot)^{\frac{1}{2}}$ and $\boldsymbol{\sigma}_{1}[\cdot]$ denotes the largest singular value.
\end{definition}

Intuitively, $\Sl_{\pm}(\mu,\chi)$ quantify the asymmetry of a state: they increase when $\chi$ becomes more prominent with $\mu$ fixed.
Yet, they do not capture all the aspects of asymmetry.
In particular, while quantifying the asymmetry in the second moment (type-2 asymmetry), they are independent of the first moment and thus fail to capture the asymmetry induced by displacements (type-1 asymmetry).
It has been found that type-1 and type-2 asymmetry are not interconvertible via GCO~\cite{Koukoulekidis2025}.
Here, we focus solely on type-2 asymmetry.

We now prove several important properties of $\Sl_{\pm}$ functions.
\begin{itemize}
	\item[(P1)]  Finiteness and faithfulness: $0\leq \Sl_\pm(\mu,\chi)<\infty$ with $\Sl_{\pm}(\mu,\chi) = 0$ if and only if $\chi = 0$. 
	\item [(P2)] Monotonicity: $\Sl_{\pm}(\mu,\chi) \geq \Sl_{\pm}(\mu',\chi')$ whenever $(\mu,\chi)\toGCO(\mu',\chi')$.
	\item [(P3)] Complete non-extensiveness: $\Sl_{\pm}(\mu_{1}\oplus\mu_{2},\chi_{1}\oplus\chi_{2}) = \max\{\Sl_{\pm}(\mu_{1},\chi_{1}),\, \Sl_{\pm}(\mu_{2},\chi_{2})\}$.
\end{itemize}
(P1--3) are proved in Proposition~\ref{proposition: kernels inclusion}, Lemmas~\ref{le:monotones_app} and~\ref{le:monotones-composite}, respectively.
Note that $(\mu_{1}\oplus\mu_{2},\chi_{1}\oplus\chi_{2})$ represents the second moment of a tensor product of two states each with the second moments $(\mu_{1},\chi_{1})$ and $(\mu_{2},\chi_{2})$.

Among the properties, the complete non-extensiveness of (P3) is rare and entails surprising consequences for state transformations. 
Importantly, it rules out any multi-copy or distillation processes achieving higher $\Sl_{\pm}$ values. 
To see this, suppose that $(\oplus_{i=1}^{N}\mu,\oplus_{i=1}^{N}\chi)\toGCO(\oplus_{i=1}^{M}\mu',\oplus_{i=1}^{M}\chi')$. 
(P2) and (P3) imply that $\Sl_{\pm}(\mu,\chi) \geq  \Sl_{\pm}(\mu',\chi')$, regardless of $N$ and $M$.
Therefore, arbitrarily many copies of a state with $(\mu,\chi)$ cannot yield even a single copy of a state with $(\mu',\chi')$ if $\Sl_{\pm}(\mu',\chi')>\Sl_{\pm}(\mu,\chi)$.
Similar no-go theorems have been established for various tasks when the allowed operations are restricted to be Gaussian, including distillation of entanglement~\cite{Eisert2002_GaussianDistillation} and squeezing~\cite{Lami2018_GaussianQRT}, quantum error correction~\cite{Niset2009_GaussianQEC}, and quantum sensing~\cite{hnhp-jhr2}.
Yet, our result on asymmetry is especially extreme:
while the marginal asymptotic distillation rate of asymmetry diverges for almost all initial states in the finite-dimensional (and thus non-Gaussian) case~\cite{Shiraishi2024_ArbitraryAmplification}, the complete non-extensiveness of $\Sl_{\pm}$ here implies that it is impossible to accumulate type-2 asymmetry across independent subsystems when the Gaussian condition is imposed.
Note that our monotones only forbid the distillation of type-2 asymmetry (concerning second moments) while type-1 asymmetry is known to be distillable via GCOs~\cite{Yadavalli2025}.

Finally, we study whether catalysts can relax state transformation conditions for GCOs.
Catalysts are auxiliary systems that participate in the process to activate otherwise prohibited transformations~\cite{JonathanP1999, LipkaBartosik2024_CatalysisReview}. 
For instance, even if the direct transformation $\rho_{S}\to\rho'_{S'}$ cannot be performed using a set of allowed operations, a transformation $\rho_{S}\otimes\tau_{C} \to \rho'_{S'}\otimes\tau_{C}$ (strict catalysis) or $\to \varrho_{S'C}$ such that $\Tr_{C}[\varrho_{S'C}] = \rho'_{S}$ and $\Tr_{S'}[\varrho_{S'C}] = \tau_{C}$ (correlated catalysis) is often achievable with an allowed operation. 
In both strict and correlated catalysis, the catalyst state $\tau_{C}$ is (marginally) preserved and can be reused for an indefinite number of times. 
Correlated catalysis, in particular, is found to relax almost all state transformation conditions in various resource theories, including the theory of athermality~\cite{Shiraishi2021_GP,Son2024hierarchy,  Shiraishi2025GPC, shiraishi2025recoverysecondlawfully}, nonuniformity~\cite{Wilming2022correlationsin}, entanglement (for pure states)~\cite{Kondra2021Entanglement}.
In the case of $\mathrm{U}(1)$-asymmetry, i.e. when covariant operations (not necessarily Gaussian) are allowed, correlated catalysis even enables arbitrary amplification of the asymmetry~\cite{Ding21_Asymmetry, Shiraishi2024_ArbitraryAmplification, Kondra2024_CoherenceAsymmetry}.
Remarkably, correlated catalysis loses its power of amplifying the type-2 asymmetry when Gaussianity is imposed. 
\begin{itemize}
	\item[(P4)] $\Sl_{\pm}(\mu,\chi)$ are monotonic under (correlated) catalytic transformations.    
\end{itemize}
A full technical statement and its proof are in Lemma~\ref{le: no catalysis}.

\emph{Discussions}---%
In this work, we studied the interplay between Gaussianity and time-translation covariance, two properties well-studied separately but rarely together.
Our main technical results are the characterization of Gaussian covariant operations.
We fully characterized the subset with a free dilation, the structure of operations with a fixed-point, and monotone functions that never increase after an operation.
We also provide a useful decomposition for operations with an asymmetric fixed-point as Lemma~\ref{le: fixed point decomposition} in Section~\ref{ap:proofs1_2}. 
These characterizations will be useful for future studies on Gaussian covariant operations, which encompass numerous important operations in optics. 
In particular, we expect the new, completely non-extensive monotones to serve as simple and powerful tools for studying general state transformation problems. 

The significance of our results is not limited to Gaussian covariant operations. 
By contrasting general time-translation‑symmetric quantum operations and thermal processes with their Gaussian subset, we obtain three main contrasts arising from Gaussianity.
First, we identify a subset of Gaussian covariant operations that cannot be dilated into any free unitary, indicating a discrepancy between axiomatically defined and operationally defined sets. 
This separation is reminiscent of similar ones that appear when additional restrictions such as Markovianity~\cite{Spaventa2022MTO, Korzekwa2022MTO} or locality~\cite{Marvian22Locality} are imposed.
Second, we prove that asymptotic or catalytic advantages vanish for single-shot restrictions given by the monotones $\Sl_{\pm}$ introduced here. 
This mirrors the impossibility of Gaussian resource distillation and broadcasting in various setups~\cite{Eisert2002_GaussianDistillation, Lami2018_GaussianQRT, Chatterjee2025}.
Finally, we demonstrate that Gaussianity can remove the separation between axiomatically and operationally defined sets by proving that Gaussian thermal operations are identical to Gaussian enhanced thermal operations. 
Such closure of the gap between thermal and enhanced thermal operations has not been observed at the channel level in other settings including those allowing catalytic channels~\cite{Son2024RC}.

What are the reasons behind these contrasts?
We outline a few possible explanations.
Given that Gaussian thermal operations are Markovian~\cite{Serafini20}, one might ask whether Markovianity is the decisive factor.
Ref.~\cite{Haagerup2011} gives a negative answer by investigating the Markovian versions of thermal operations and enhanced thermal operations and explicitly constructing a counterexample showing the gap between the two. 
This is a clear sign that other aspects of Gaussianity are more crucial.
Another physical intuition is that Gaussian operations, generated by quadratic Hamiltonian, decompose into two-mode interactions, unlike general symmetric operations~\cite{Marvian22Locality}.
This locality might underlie the absence of asymptotic/catalytic advantages and the absence of separation between GEnTO and GTO.
Further work on the complete characterization of asymptotic and catalytic transformations beyond our single‑mode and second‑moment results would help test this intuition.

We finally point out two promising avenues for future work. 
First is the calculation of approximate implementation cost of GCOs.
We established that the exact implementation cost diverges for some GCOs including amplifiers, yet they are used in practice. 
It means that with sufficiently large cost, GCOs without free dilation can be approximated well. 
Determining the trade-off between approximation accuracy and implementation cost would impact theorists and practitioners alike.
Second is the better understanding of $\Sl_{\pm}$ monotones. 
Definition~\ref{def: Sl monotones} suggests that a single mode, corresponding to the vector achieving the optimal ratio of $\chi$ and $\mu$, determines the monotone values, but it is not clear what this mode is for general multi-mode systems. 
We anticipate that better insight into this optimal mode should illuminate hidden structure in Gaussian states.

\emph{Acknowledgements}---%
The work of the NTU team is supported by the start-up grant of the Nanyang Assistant Professorship at the Nanyang Technological University in Singapore, and by the National Research Foundation, Singapore through the National Quantum Office, hosted in A*STAR, under its Centre for Quantum Technologies Funding Initiative (S24Q2d0009). 
The work of the Ulm team is supported by the DFG via QuantERA project ExTRaQT (grant no. 499241080) and the ERC Synergy grant HyperQ (grant no. 856432). 
This project was initiated as a result of discussions during the conference Quantum Resources 2023 held in Singapore.

\bibliography{ref.bib}

\newpage

\normalsize
\renewcommand{\theequation}{A\arabic{equation}}
\setcounter{equation}{0}  

\section*{End Matter}

\emph{Appendix A: Single-mode state transformations}---%
It is instructive to consider the simplest case of single-mode bosonic systems.
Note that the second moments $\mu$ and $\chi$, as well as the GCO matrices $A$ and $B$ are numbers in this case.
For a single-mode second moment $(\mu,\chi)$, our monotones are given as $\Sl_{\pm}(\mu,\chi) = \frac{|\chi|}{\mu \pm \frac{1}{2}}$ when $\mu>\frac{1}{2}$ and $\Sl_{\pm}(\mu,\chi) = 0$ when $\mu=\frac{1}{2}$, corresponding to the second moment of a vacuum state. 
In Figure~\ref{fig:single_mode_cone}, single-mode states are visualized in a two-dimensional plot where $\Sl_{-}$ corresponds to the slope of the line connecting the state of interest and the single mode vacuum, whereas $\Sl_{+}$ corresponds to the other slope from the $(-\frac{1}{2},0)$ point. 

Furthermore, the monotones $\Sl_{\pm}$ also become sufficient for the second moment transformation, i.e. they become a pair of complete monotones for the second moments under GCO. 
\begin{lemma}
	\label{le:single_mode_state}
	Suppose that second moments $(\mu,\chi)$ and $(\mu',\chi')$ represent a single-mode system.
    Then, $(\mu,\chi)\toGCO(\mu',\chi')$ if and only if 
	\begin{align}
		\Sl_{\pm}(\mu,\chi)\geq \Sl_{\pm}(\mu',\chi').\label{eq:slope}
	\end{align}
\end{lemma}
\begin{proof}
	The `only if' part follows directly from monotonicity of $\Sl_{\pm}$.
	To prove the `if' part, we construct a GCO converting $(\mu,\chi)$ to $(\mu',\chi')$ when monotonicity is satisfied.
	
	Recall that the monotones $\Sl_{\pm}$ are always non-negative from (P1).
	If $\chi = 0$, the monotone $0 = \Sl_{\pm}(\mu,\chi) \geq \Sl_{\pm}(\mu',\chi')$; which means that $\chi' = 0$ must hold.
	Let
	\begin{align}
		a_{\pm}=\sqrt{\frac{\mu'\pm\frac{1}{2}}{\mu\pm\frac{1}{2}}},\ b_{\pm}=\pm\frac{1}{2}(a^2_\pm-1).
	\end{align}
	For $\mu'\geq \mu$, which include the case where $\mu=\frac{1}{2}$, $(\mu,0)$ can be transformed to $(\mu',0)$ by the GCO $(A,B)=(a_+,b_+)$; For $\mu'< \mu$, which implies $\mu>\frac{1}{2}$, the GCO $(A,B)=(a_-,b_-)$ transforms $(\mu,0)$ to $(\mu',0)$.
	The constraint $B\geq\pm\frac{1}{2}(1-|A|^2)$ is satisfied for both cases.
	
	If $\chi\neq0$, the second moment transformation by a GCO $(A,B)$ reads
	\begin{align}
		\mu' &= |A|^{2}\mu + B,\\
		\chi' &= A^{*2}\chi,
	\end{align}
	which can be achieved by $A=\sqrt{\chi^{\prime*}/\chi^*}$ and $B=\mu'-\mu \cdot |\chi'/\chi|$.
	Since
	\begin{align}
		B\pm\frac{1}{2}(1-|A|^2)
		&=  (\mu'\pm\frac{1}{2})-\left|\frac{\chi'}{\chi}\right|(\mu\pm\frac{1}{2})\\
		&=  (\mu'\pm\frac{1}{2})\left[1-\frac{\Sl_{\pm}(\mu',\chi')}{\Sl_{\pm}(\mu,\chi)}\right]\geq0.\nonumber
	\end{align}
	$(A,B)$ constitutes a proper GCO.
\end{proof}

The reachable single-mode states following Lemma~\ref{le:single_mode_state} can be represented as the blue colored region in Fig.~\ref{fig:single_mode_cone}.

\begin{figure}[h]
	\centering
	\includegraphics[width=\linewidth]{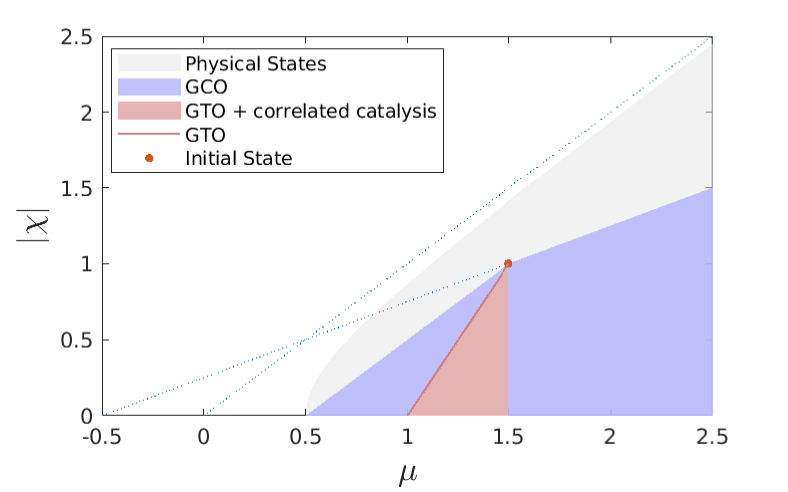}
	\caption{Single-mode state transformations for the second moment $(\mu,\chi)$.
		Colored parts represent the reachable states from a given initial state (red dot) via GCO (purple region), GTO (red line), and GTO with a correlated catalysis.}
	\label{fig:single_mode_cone}
\end{figure}

The completeness of $\Sl_{\pm}$ monotones also implies that correlated catalysts (however big and complicated they are) cannot offer any advantage for single-mode transformations. 
\begin{theorem}[No-catalysis]\label{thm: no catalysis} 
	Let $S$ be a single-mode system and $C$ be an $m$-mode catalyst system. 
    Suppose that the second moment of $S$ can be transformed as $(\mu_S,\chi_S)\to(\mu_S',\chi_S')$ with the help of a catalytic auxiliary state. 
    In other words, the transformation is possible via a joint GCO on $SC$ that preserves the second moment of the system $C$.
    Then, there is a GCO acting locally on $S$ achieving $(\mu_S,\chi_S)\toGCO(\mu_S',\chi_S')$. 
\end{theorem}
\begin{proof}
	The monotonicity of $\Sl_{\pm}$ under correlated catalytic transformation (P4) combined with the completeness of these functions (Lemma~\ref{le:single_mode_state}) implies the theorem. 
\end{proof}

\clearpage
\newpage
\setcounter{page}{1}
\title{Supplemental Materials for ``Time-translation symmetry and thermality in Gaussian operations"}
\maketitle
\onecolumngrid

\setcounter{secnumdepth}{3}
\setcounter{section}{0}
\renewcommand{\thesection}{S\arabic{section}}
\renewcommand{\thetheorem}{S.\arabic{theorem}}
\setcounter{theorem}{0}
\renewcommand{\thelemma}{S.\arabic{lemma}}
\setcounter{lemma}{0}
\renewcommand{\thecorollary}{S.\arabic{corollary}}
\setcounter{corollary}{0}
\renewcommand{\theequation}{S\arabic{equation}}
\setcounter{equation}{0} 
\renewcommand{\theremark}{S\arabic{remark}}
\setcounter{remark}{0} 
\renewcommand{\theproposition}{S\arabic{proposition}}
\setcounter{proposition}{0}

\section{Full characterization of Gaussian Covariant Channels}\label{ap:proofs1}
In this section we scrutinize the full characterization of Gaussian covariant channels. We begin with full generality, and first show in Proposition~\ref{pro: indep evols} that modes with non-resonant frequencies can be considered independently of each other. 

\begin{proposition}\label{pro: indep evols}
	Let $F^{\omega}(S)$ be a subsystem of $S$ consisting of all bosonic modes with frequency $\omega$. 
	Suppose that $\mE_{\mathrm{GCO}}$ is a GCO from an $\Inm$ modes system $I$ to an $\Outm$ modes system $O$. 
	Then we can always write $\mE_{\mathrm{GCO}} = \mE_{3}\circ\mE_{2}\circ\mE_{1}$, i.e. as a concatenation of three GCOs:
	\begin{enumerate}
		\item $\mE_{1} = \Tr_{\{F^{\omega}(I)\,|\, F^{\omega}(O)=\emptyset\}}$ tracing out all input modes that do not have the output modes with the same frequency,
		\item $\mE_{2} = \bigoplus_{\{\omega \,|\, F^{\omega}(I),F^{\omega}(O) \neq \emptyset\}} \mE_{\omega}$, where each $\mE_{\omega}$ is a GCO from $F^{\omega}(I)$ to $F^{\omega}(O)$ when both are non-empty, and
		\item $\mE_{3}(\cdot) = \cdot\bigotimes_{\{F^{\omega}(O)\,|\, F^{\omega}(I)=\emptyset\}} \sigma_{F^{\omega}(O)}$ appending a fixed symmetric state $\sigma_{F^{\omega}(O)}$ that is independent of the input state for each output modes that do not have the input modes with the same frequency.
	\end{enumerate}
\end{proposition}
\begin{proof}
	Suppose that $\mE_{\mathrm{GCO}}$ is a GCO, i.e. it satisfies the covariance condition $\mE_{\mathrm{GCO}}\circ\mathcal{U}_t^\mathrm{in} = \mathcal{U}_t^\mathrm{out}\circ\mE_{\mathrm{GCO}}$ with time evolution operations $\mathcal{U}_{t}^{\mathrm{in/out}}$. 
	Note that the unitary processes $\mathcal{U}_t^\mathrm{in}$ and $\mathcal{U}_t^\mathrm{out}$ are associated with the tuples $(\tR_{\mathrm{in}}(t),0,0)$ and $(\tR_{\mathrm{out}}(t),0,0)$, where
	\begin{align}
		\tR_{\mathrm{in}}(t) = \left(\bigoplus_{i} e^{-i\omega_{i}t}\right)\oplus \left(\bigoplus_{i} e^{i\omega_{i}t}\right),\quad \tR_{\mathrm{out}}(t) = \left(\bigoplus_{i} e^{-i\omega'_{i}t}\right)\oplus \left(\bigoplus_{i} e^{i\omega'_{i}t}\right),
	\end{align}
	and $\{\omega_{i}\}_{i}$ and $\{\omega'_{i}\}_{i}$ are input/output frequencies.
	Then $\mE_{\mathrm{GCO}}$ has an one-to-one correspondence with a tuple $(\tX,\tY,\td)$ satisfying
	\begin{align}
		\tX\tR_{\mathrm{in}}(t) \vecr + \td & =  \tR_{\mathrm{out}}(t)\tX \vecr+\tR_{\mathrm{out}}(t)\td, \label{eq:c1} \\
		\tX\tR_{\mathrm{in}}(t)  M \tR_{\mathrm{in}}^{\dagger}(t)\tX^\dagger+\tY & = \tR_{\mathrm{out}}(t)\tX  M \tX^\dagger\tR_{\mathrm{out}}^{\dagger}(t)+\tR_{\mathrm{out}}(t)\tY\tR_{\mathrm{out}}^{\dagger}(t),\label{eq:c2}
	\end{align}
	for all $t,\vecr,M$. 
	It follows that $\tX\tR_{\mathrm{in}}(t) = \tR_{\mathrm{out}}(t)\tX$ from Eq.~\eqref{eq:c1} and this combined with Eq.~\eqref{eq:c2} gives $\tY = \tR_{\mathrm{out}}(t)\tY\tR_{\mathrm{out}}^{\dagger}(t)$ for all $t$.
	The last condition implies that $\tY$ is block-diagonal where each block corresponds to each output frequencies $F^{\omega}(O)$. Meanwhile, $\tX\tR_{\mathrm{in}}(t) = \tR_{\mathrm{out}}(t)\tX$ gives us
	\begin{align}
		X_{i,j}\omega_{j} = \omega'_{i}X_{i,j},\quad X_{i+\Outm,j+\Inm}\omega_{j} = \omega'_{i}X_{i+\Outm,j+\Inm},\quad X_{i,j+\Inm} = X_{i+\Outm,j} = 0, 
	\end{align}
	for all $i = 1,\cdots,\Outm$, and $j = 1,\cdots,\Inm$.
	Hence, the $i$th output mode after the operation $\mE_{\mathrm{GCO}}$ is independent of the $j$th input mode, unless $\omega_{j} = \omega'_{i}$.
	This necessarily gives the structure of the channel as stated in the proposition. 
\end{proof}

Proposition~\ref{pro: indep evols} allows us to then focus on cases where all input and output frequencies are the same. 

\begin{lemma}
\label{le:gene_GCO_app}
    Any $\Inm$-mode to $\Outm$-mode GCO is fully characterized by a pair $(A,B)$ of $\Outm\times \Inm$ matrix $A$ and $\Outm\times \Outm$ matrix $B$, transforming first and second moments as: 
    \begin{align}\label{eq:Gaussian_characterization}
        \vecalpha \rightarrow A\vecalpha, \qquad 
        \mu \rightarrow A\mu A^\dagger + B, \qquad
        \chi \rightarrow A^* \chi A^\dagger,
    \end{align}
    with the constraint
    \begin{align}
    	B\geq \pm\frac{1}{2}(  \iden-AA^\dagger).\label{eq:HUR_A_B_app}
	\end{align}
\end{lemma}
\begin{proof}
	Let $S$ be an $\Inm$-mode system and $S'$ be an $\Outm$-mode system. 
	As we are assuming all input and output mode frequencies are the same, we denote $\tR_{\mathrm{in}}(\theta)= \cos(\theta) \iden_{S}-i\sin(\theta) \cdot 2\mathbf{Z}_{S}$ and $\tR_{\mathrm{out}}(\theta)= \cos(\theta) \iden_{S'}-i\sin(\theta) \cdot 2\mathbf{Z}_{S'}$ for $\theta$ phase shifts in $S$ and $S'$, respectively.
	By definition, a GCO tuple $(\tX,\tY,\td)$  satisfies
	\begin{align}
		\tX\tR_{\mathrm{in}}(\theta) \vecr + \td & =  \tR_{\mathrm{out}}(\theta)\tX \vecr+\tR_{\mathrm{out}}(\theta)\td,\\
		\tX\tR_{\mathrm{in}}(\theta)  M \tR^{\dagger}_{\mathrm{in}}(\theta)\tX^\dagger+\tY & =  \tR_{\mathrm{out}}(\theta)\tX  M \tX^\dagger\tR_{\mathrm{out}}^{\dagger}(\theta)+\tR_{\mathrm{out}}(\theta)\tY\tR_{\mathrm{out}}^{\dagger}(\theta),
	\end{align}
	for all $M$, $\vecr$, and $\theta$. It follows that $\td=0$, $\mmZ_{S'}\tX\mmZ_{S}=\tX$, and $\mmZ_{S'}\tY\mmZ_{S'}=\tY$. Recall, furthermore, that $M$ has a block-diagonal structure as specified in Eq.~\eqref{eq:cvm}. 
	Therefore, $\tX, \tY$ admit the block forms:
	\begin{align}
		\tX=\begin{pNiceArray}{c:c}
			A^* & 0 \\
			\hdottedline
			0 & A
		\end{pNiceArray}, \qquad \ 
		\tY=\begin{pNiceArray}{c:c}
			B^* & 0 \\
			\hdottedline
			0 & B
		\end{pNiceArray},
	\end{align}
	with an $\Outm\times \Inm$ matrix $A$ and an $\Outm\times \Outm$ matrix $B$. 
	Then the constraint $\tY \geq \tX \mathbf{Z}_{S} \tX^\dagger - \mathbf{Z}_{S'}$ translates to
	\begin{align}
		\begin{pNiceArray}{cc}
			B^* & 0 \\
			0 & B
		\end{pNiceArray}\geq \frac{1}{2}\begin{pNiceArray}{cc}
			A^* & 0 \\
			0 & A
		\end{pNiceArray}\begin{pNiceArray}{cc}
			\iden_{S} & 0 \\
			0 & -\iden_{S}
		\end{pNiceArray}\begin{pNiceArray}{cc}
			A^{*\dagger} & 0 \\
			0 & A^\dagger
		\end{pNiceArray}-\frac{1}{2}\begin{pNiceArray}{cc}
			\iden_{S'} & 0 \\
			0 & -\iden_{S'}
		\end{pNiceArray},
	\end{align}
	which is equivalent to Eq.~\eqref{eq:HUR_A_B_app}.
	The state transformation under a generic GCO $(A,B)$ is then described as Eqs.~\eqref{eq:Gaussian_characterization}.
\end{proof}

Any GCO $(A,B)$ with $\Inm\neq\Outm$ can be realized by a GCO $(A',B')$ with $\Inm'=\Outm'=\max\{\Inm,\Outm\}$ for the following reason. 
If $\Inm>\Outm$, we construct $(A',B')$ as $A'=\begin{pmatrix} A\\0 \end{pmatrix},B'=B\oplus\frac{\iden_{\Inm-\Outm}}{2}$. Then, $(A,B)$ is equivalent to $(A',B')$ followed by removing $(\Inm-\Outm)$ vacuum modes from the output. 
Similarly, if $\Inm<\Outm$, we set $A'=\begin{pmatrix} A & 0 \end{pmatrix},B'=B$. 
Then $(A,B)$ is equivalent to adding $(\Outm-\Inm)$ modes in a free state $(\mu_v,0)$ to the input followed by $(A',B')$. 
Therefore, in the subsequent section, we only need to focus on the GCOs with $\Inm=\Outm\equiv m$.

\section{Full characterization of freely dilatable Gaussian covariant channels}\label{ap:proofs1_2}
We restate and prove Lemma~\ref{le:free_dilation} of the main text, which is our central technical result and sets a foundation for all the subsequent main results such as Theorem~\ref{th:gento_gto}, Lemma~\ref{le:single_mode_state} and Theorem~\ref{thm: no catalysis}. 
We restate the lemma here as Lemma~\ref{le:free_dilation_app} for the reader's convenience.
\begin{lemma}\label{le:free_dilation_app}
A GCO $(A,B)$ is freely dilatable if and only if the following two conditions are satisfied.
\begin{itemize}
    \item[(F1)] $\iden-AA^\dagger\geq0$;
    \item[(F2)] $\supp(B)= \supp(\iden-AA^\dagger)$.
\end{itemize}
Furthermore, when (F1) or (F2) is not satisfied, even using a non-Gaussian and/or asymmetric auxiliary state does not enable the dilation with a Gaussian covariant unitary. 
\end{lemma}
\begin{proof}
    If the two conditions are satisfied, the dilation in Definition~\ref{def: free dilation} can be constructed as follows. 
    Let
    \begin{align}
        A=U\begin{pNiceArray}{cc}
            \sqrt{1-\Lambda_+} & 0 \\
            0 & \iden_{m-m_+}
        \end{pNiceArray}W,\label{eq:A_SVD}
    \end{align}
    be the singular value decomposition of $A$, where $\Lambda_+=\mathrm{diag}(\lambda_1,\dots,\lambda_{m_+})>0$ is a $m_+$-dimensional diagonal matrix, $m_+\leq m$, and $U,W$ are $m\times m$-dimensional unitaries.
    Then the spectral decomposition of  $\iden-AA^\dagger$ is 
	\begin{align}
	    \iden-AA^\dagger = U\begin{pNiceArray}{cc}
		        \Lambda_+ & 0 \\
		        0 & 0
		    \end{pNiceArray}U^\dagger\label{eq:I_AA_SVD},
	\end{align}
	and (F2) is then equivalent to $B$ having the form
	\begin{align}
	    B=U\begin{pNiceArray}{cc}
		        B_+ & 0 \\
		        0 & 0
		    \end{pNiceArray}U^\dagger,\label{eq:B_decom}
	\end{align}
	where $B_+>0$ is an $m_+$-dimensional completely positive matrix. 
   A straightforward calculation shows that an $m_+$-mode auxiliary state with
    \begin{align}
        \mu_R=\Lambda_+^{-\frac{1}{2}}B_+\Lambda_+^{-\frac{1}{2}},\ \qquad \chi_R=0,
    \end{align}
    and a joint unitary
    \NiceMatrixOptions{cell-space-limits = 2pt}
    \begin{align}
        V=\begin{pNiceArray}{c:c}
            U & 0 \\
            \hdottedline
            0 & \iden_{m_+}
        \end{pNiceArray}
        \begin{pNiceArray}{cc:c}
            \sqrt{1-\Lambda_+} & 0 & \sqrt{\Lambda_+} \\
            0 & \iden_{m-m_+} & 0 \\
            \hdottedline
            -\sqrt{\Lambda_+} & 0 & \sqrt{1-\Lambda_+} 
        \end{pNiceArray}
        \begin{pNiceArray}{c:c}
            W & 0 \\
            \hdottedline
            0 & \iden_{m_+}
        \end{pNiceArray}\label{eq:dila_V}
    \end{align}
	gives the desired dilation. Besides, the uncertainty relation $M_R+\mmZ\geq0$, which reduces to $\mu_R\geq\frac{\iden}{2}$, is ensured by the constraint $B\geq\frac{1}{2}(\iden-AA^\dagger)$.

    Conversely, if a GCO $(A,B)$ has a free dilation $(\mu_R,0)$ and
    \begin{align}
        V=\begin{pNiceArray}{c:c}
            V_{SS} & V_{SR} \\
            \hdottedline
            V_{RS} & V_{RR}
        \end{pNiceArray},
    \end{align}
    then
    \begin{align}
        V_{SS} &=A,\label{eq:vss}\\
        V_{SR}\mu_R V_{SR}^\dagger & = B.\label{eq:vsr}
    \end{align}
    Because $\iden-V_{SS}V_{SS}^\dagger=V_{SR}V_{SR}^\dagger\geq0$, Eq. (\ref{eq:vss}) holds only if condition (F1) is satisfied. Further, because
    \begin{align}
        V_{SR}V_{SR}^\dagger=\iden-AA^\dagger=U
        \begin{pNiceArray}{cc}
        \Lambda_+ & 0 \\
        0 & 0
    \end{pNiceArray}U^\dagger,
    \end{align}
    the SVD of $V_{SR}$ reads
    \begin{align}
        V_{SR}=U
        \begin{pNiceArray}{cc}
        \sqrt{\Lambda_+} & 0 \\
        0 & 0
    \end{pNiceArray}
    U_R.
    \end{align}
    Substituting it to Eq.~\eqref{eq:vsr}, we obtain
    \begin{align}
        B= U\begin{pNiceArray}{cc}
        \sqrt{\Lambda_+} & 0 \\
        0 & 0
    \end{pNiceArray}U_R\mu_R U_R^\dagger\begin{pNiceArray}{cc}
        \sqrt{\Lambda_+} & 0 \\
        0 & 0
    \end{pNiceArray} U^\dagger,
    \end{align}
    which is exactly the form of Eq.~\eqref{eq:B_decom}.

    Finally we observe that even if we change the auxiliary state to be a generic quanutm state with the second moment $(\mu_{R},\chi_{R})$ and the first moment $\vecalpha_{R}$, their action on the second moment of the system state does not change. 
    Therefore, adding non-Gaussian or asymmetric resource to the auxiliary state does not help implementing GCOs without free dilations.
\end{proof}

The next two results, Theorem \ref{th:GCO_PI_app} (a reiteration of Theorem \ref{th:GCO_PI} in the main text) and Lemma~\ref{le:free_dilation_app} relate to what properties we can infer from a Gaussian covariant operation that admits a particular fixed point dilation.

\begin{theorem}
\label{th:GCO_PI_app}
	If a GCO mapping a system to itself has no fixed-point that is a valid quantum state, then it is not freely dilatable. 
	However, the converse is not true.
\end{theorem}
\begin{proof}
	We prove the first claim by contrapositive. 
	We explicitly construct a valid fixed-point for any $m$-mode to $m$-mode GCO $(A,B)$ with a free dilation.
	Define a matrix 
	\begin{align}
		\Delta \coloneqq B-\frac{1}{2}(\iden-AA^\dagger) \geq 0,
	\end{align}
	where the positive semidefiniteness comes from the CP constraint $B \geq \frac{1}{2}(\iden - AA^{\dagger})$.
	Next, we define 
	\begin{align}
		\mu_{\star} \coloneqq \frac{\iden}{2} + \sum_{n=0}^{\infty} A^{n} \Delta (A^\dagger)^{n}.
	\end{align}
	We show that $(\mu_{\star},\chi_{\star} = 0)$ can be a valid choice of the second moment for a quantum state using the following arguments.
	\begin{enumerate}
        \item A direct calculation shows $\mu_\star=A\mu_\star A^\dagger+B$.
		\item $A^{n} \Delta (A^\dagger)^{n}\geq0$ for any $n\geq0$ since $\Delta\geq0$.
		Thus, $\mu_{\star} \geq \frac{\iden}{2}$.
		\item From Lemma~\ref{le:free_dilation_app}, we have $AA^\dagger\leq \iden$ and $\supp(\Delta)\subseteq\supp(\iden-AA^\dagger)$ meaning that $\Delta$ has support only on the strict-contraction subspace of $A$. 
		Hence, the series $\sum_{n=0}^{\infty} A^{n} \Delta (A^\dagger)^{n}$ converges and $\mu_{\star}$ is finite. 
	\end{enumerate}
	Therefore the centered symmetric Gaussian state with moments $(\vec\alpha_\star=0,\mu_\star,\chi_\star=0)$ is a valid fixed point.
	
	For the second claim, we provide a counterexample: a GCO with a fixed-point but admits no free dilation.
	Consider a two-mode GCO with
	\begin{align}
		A=\begin{pNiceArray}{cc}
			0 & \sqrt{\eta_1} \\
			\sqrt{\eta_2} & 0
		\end{pNiceArray},\quad 
		B=\begin{pNiceArray}{cc}
			b_{11} & 0 \\
			0 & b_{22}
		\end{pNiceArray},
	\end{align}
	where $\eta_1\eta_2<1$, $\eta_2>1$, and $b_{11}\geq\frac{1}{2}(1-\eta_1)$, $b_{22}\geq\frac{1}{2}(\eta_2-1)$.
	The fixed-point of such a channel is
	$\mu=\begin{pNiceArray}{cc}
		\mu_{11} & 0 \\
		0 & \mu_{22}
	\end{pNiceArray}$ and $\chi=0$, where
	\begin{align}
		\mu_{11}=\frac{b_{11}+\eta_1 b_{22}}{1-\eta_1\eta_2},\quad 
		\mu_{22}=\frac{b_{22}+\eta_2 b_{11}}{1-\eta_1\eta_2}.
	\end{align}
	However, this channel does not admit a physical implementation, because $\eta_2>1$ violates Condition (F1) in Lemma~\ref{le:free_dilation_app}.
\end{proof}

Here we present a more detailed and rigorous version of the physical intuition given in the main text that explains this no‑go result.
First we prove by contradiction that any GCO without a fixed-point quantum state must be capable of increasing the largest eigenvalue of the input $\mu$ matrix of the second moment (corresponding to the energy of the hottest normal mode), irrespective of how large it already is.
Define $\mathcal{F}_{\bar{n}_0}$ as the set of all $m$-mode second moments $\mu$ whose largest eigenvalue (denoted $\nu^{\downarrow}_{1}(\mu)+\frac{1}{2}$) is upper bounded by $\bar{n}_{0}+\frac{1}{2}$, i.e. $\nu^{\downarrow}_{1}(\mu)\leq\bar{n}_{0}$.
It is straightforward to check that $\mathcal{F}_{\bar{n}_0}$ is a nonempty, compact, convex set for any fixed $\bar{n}_{0}\geq 0$.
Recall that a GCO linearly transforms the second moment as $\mu \mapsto A\mu A^{\dagger} +B$.
Suppose a GCO cannot increase the largest eigenvalue $\nu^{\downarrow}_{1}(\mu)$ when it is already large; more precisely, suppose that the inequality $\nu^{\downarrow}_{1}(\mu') \leq \nu^{\downarrow}_{1}(\mu)$ holds for any input second moment $\mu$, output $\mu'$, whenever $ \nu^{\downarrow}_{1}(\mu)\geq \bar{n}_{0} $ for some fixed number $\bar{n}_{0}\geq0$.
Then the transformation corresponding to this GCO preserves the set $\mathcal{F}_{\bar{n}_0}$.
By the fixed-point theorem, this transformation has a fixed second moment in $\mathcal{F}_{\bar{n}_0}$.
Let $\mu_{0}$ be such a fixed second moment: then the Gaussian symmetric state with second moment $\mu_{0}$ and zero first moment is then a fixed-point quantum state, because GCOs cannot make a state non-Gaussian nor create a nonzero first moment.
This contradicts the assumption that the GCO has no fixed‑point quantum state.

Now we show that a GCO with a free dilation cannot increase the energy of the hottest mode when it is already too high. 
More precisely: a GCO that admits a free dilation always has a bound $\bar{n}_{0}\geq0$ such that $\nu^{\downarrow}_{1}(\mu') \leq \nu^{\downarrow}_{1}(\mu)$ for any input and output second moments $\mu, \mu'$ whenever $ \nu^{\downarrow}_{1}(\mu)\geq \bar{n}_{0} $.
When a GCO admits a free dilation, it can be regarded as a physical process of appending a Gaussian symmetric ancilla, applying passive linear unitaries, and then tracing out the ancilla.
In terms of second moment transformations this corresponds to i) $\mu \mapsto \mu\oplus\mu_{A}$ with the ancilla second moment $\mu_{A}$, ii) $\mu\oplus\mu_{A} \mapsto V\mu\oplus\mu_{A}V^{\dagger}$ for some unitary matrix $V$, and iii) taking an $m\times m$ submatrix $\mu'$ of $V\mu\oplus\mu_{A}V^{\dagger}$. 
Since the largest eigenvalue of any submatrix is bounded by that of the entire matrix, we find that $\nu^{\downarrow}_{1}(\mu') \leq \max\{\nu^{\downarrow}_{1}(\mu) , \nu^{\downarrow}_{1}(\mu_{A}) \}$.
Choosing $\bar{n}_{0}$ to be $\nu^{\downarrow}_{1}(\mu_{A}) < \bar{n}_{0}$ then implies this GCO preserves the set $\mathcal{F}_{\bar{n}_0}$ as desired.
Therefore it is impossible for a GCO to admit a free dilation but no fixed-point quantum state.

\newcommand{\mfA}{\mathfrak{A}}
\newcommand{\mfa}{\mathfrak{a}}
\newcommand{\mfB}{\mathfrak{B}}
\newcommand{\mfb}{\mathfrak{b}}
\newcommand{\mfX}{\mathfrak{X}}
\newcommand{\mfD}{\mathfrak{D}}
\newcommand{\mfd}{\mathfrak{d}_{11}}

Lastly, we present a powerful technical lemma that helps simplify much of our analysis.
This lemma decomposes a Gaussian covariant channel into two parts.
The parts depend on the fixed-point of the GCO, in particular the off-diagonal part of the fixed-point $\chi$.
This lemma tells us that one can decompose the GCO---up to some unitary---as an independent evolutions of the support and kernel of $\chi$.
We will later use it to prove Lemma \ref{le: catalyst-vac-removal}.

\begin{lemma}\label{le: fixed point decomposition}
	Suppose that an $m$-mode second moment $(\mu,\chi)$ with $\chi\neq0$ is a fixed-point of a GCO $(A,B)$. 
	Then there exists a unitary operator $U$ such that 
	\begin{align}
		U^{*}\chi U^{\dagger} = \hat{\chi}_{11} \oplus 0,\quad	UAU^{\dagger} = \hat{a}_{11} \oplus \hat{a}_{22},\quad UBU^{\dagger} = 0 \oplus \hat{b}_{22},
	\end{align}
	where the decomposition $X\oplus Y$ is such that $X$ acts on the support of $U^{*}\chi U^{\dagger}$, and $Y$ on the kernel of $U^{*}\chi U^{\dagger}$, respectively. 
	Furthermore, $\hat{\chi}_{11}>0$ and $\hat{a}_{11}$ is a unitary matrix. 
	
	Physically, it implies that a GCO $(A,B)$ with an asymmetric fixed-point $(\mu,\chi)$ can always be decomposed into three steps: i) a unitary GCO $(U,0)$, ii) a GCO $(\hat{a}_{11}\oplus \hat{a}_{22},0\oplus \hat{b}_{22})$ with the unitary $\hat{a}_{11}$, and iii) a unitary GCO $(U^{\dagger},0)$.
\end{lemma}
\begin{proof}
	The fixed point condition translates to 
	\begin{align}
		\mu &= A\mu A^{\dagger} + B,\label{eq: mu invariant}\\
		\chi &= A^{*}\chi A^{\dagger}.\label{eq: chi invariant}
	\end{align}
	We define another set of matrices
	\begin{align}\label{eq: mf var definitions}
		\mfA \coloneqq \mu^{-\frac{1}{2}}A\mu^{\frac{1}{2}},\quad \mfB \coloneqq \mu^{-\frac{1}{2}}B\mu^{-\frac{1}{2}},\quad \mfX \coloneqq (\mu^{*})^{-\frac{1}{2}}\chi\mu^{-\frac{1}{2}},
	\end{align}
	using that $\mu>0$. It follows from $B\geq0$ that $\mfB\geq0$.
	Eqs.~\eqref{eq: mu invariant} and \eqref{eq: chi invariant} 
    become
	\begin{align}
		\iden &= \mfA\mfA^{\dagger} + \mfB,\label{eq: mfA and mfB}\\
		\mfX &= \mfA^{*}\mfX\mfA^{\dagger}.\label{eq: mfX invariant}
	\end{align}
	
	$\mfX$ is a symmetric matrix and thus admits the Autonne-Takagi factorization, i.e. there exists an $m\times m$-unitary matrix $W$ such that $\mfX = W^{\tp}\mfD W$, where $\mfD$ is a diagonal matrix such that $\mfD = \mfd \oplus 0$ and $\mfd>0$ is a $m'\times m'$-diagonal matrix.
	Now Eq.~\eqref{eq: mfX invariant} can be rewritten as $W^{\tp}\mfD W = \mfA^{*}W^{\tp}\mfD W\mfA^{\dagger}$, or equivalently,
	\begin{align}
		\mfD = W^{*}\mfA^{*}W^{\tp} \mfD W\mfA^{\dagger}W^{\dagger} \eqqcolon \tilde{\mfA}^{*} \mfD \tilde{\mfA}^{\dagger},\label{eq: chi invariant alt}
	\end{align}
	where $\tilde{\mfA} \coloneqq W\mfA W^{\dagger} $.
	Similarly, we can define $\tilde{\mfB} \coloneqq W\mfB W^{\dagger}$, and Eq.~\eqref{eq: mfA and mfB} becomes
	\begin{align}
		\iden = \tilde{\mfA}\tilde{\mfA}^{\dagger} + \tilde{\mfB}.\label{eq: mu invariant alt}
	\end{align}
	Now we write $\tilde{\mfA}$ and $\tilde{\mfB}$ in the block form
	\begin{align}
		\tilde{\mfA} = \begin{pNiceArray}[cell-space-limits = 3pt]{c:c}
			\tilde{\mfa}_{11} & \tilde{\mfa}_{12}\\
			\hdottedline
			\tilde{\mfa}_{21} & \tilde{\mfa}_{22}
		\end{pNiceArray},\quad \tilde{\mfB} = \begin{pNiceArray}[cell-space-limits = 3pt]{c:c}
			\tilde{\mfb}_{11} & \tilde{\mfb}_{12}\\
			\hdottedline
			\tilde{\mfb}_{21} & \tilde{\mfb}_{22}
		\end{pNiceArray},
	\end{align}
	where the blocks $\tilde{\mfa}_{11}$ and $\tilde{\mfb}_{11}$ act on the same space as $\mfd$. 
	Using the block form, Eq.~\eqref{eq: chi invariant alt} is equivalent to 
	\begin{align}
		\tilde{\mfa}_{11}^{*}\mfd\tilde{\mfa}_{11}^{\dagger} &= \mfd,\label{eq: A11 equation}\\
		\tilde{\mfa}_{11}^{*}\mfd\tilde{\mfa}_{21}^{\dagger} &= 0,\quad
		\tilde{\mfa}_{21}^{*}\mfd\tilde{\mfa}_{11}^{\dagger} = 0,\quad 
		\tilde{\mfa}_{21}^{*}\mfd\tilde{\mfa}_{21}^{\dagger} = 0,\label{eq: A21 zero}
	\end{align}
	and Eq.~\eqref{eq: A21 zero} combined with $\mfd>0$ indicates $\tilde{\mfa}_{21} = 0$.
	Furthermore, Eq.~\eqref{eq: A11 equation} indicates that $|\det(\mfa_{11})| = 1$.
	Similarly, Eq.~\eqref{eq: mu invariant alt} becomes 
	\begin{align}
		\iden &= \tilde{\mfa}_{11}\tilde{\mfa}_{11}^{\dagger} + \tilde{\mfa}_{12}\tilde{\mfa}_{12}^{\dagger} + \tilde{\mfb}_{11},\label{eq: A11 equation 1}\\
		0 &= \tilde{\mfa}_{12}\tilde{\mfa}_{22}^{\dagger} + \tilde{\mfb}_{12},\label{eq: A12 equation 1}\\
		0 &= \tilde{\mfa}_{22}\tilde{\mfa}_{12}^{\dagger} + \tilde{\mfb}_{21},\label{eq: A21 equation 1}\\
		\iden &= \tilde{\mfa}_{22}\tilde{\mfa}_{22}^{\dagger} + \tilde{\mfb}_{22}.\label{eq: A22 equation 1}
	\end{align}
	We have $\tilde{\mfa}_{12}\tilde{\mfa}_{12}^{\dagger}\geq0$ and $\tilde{\mfb}_{11}\geq0$, which is implied from $\tilde{\mfb}\geq0$.
	Eq.~\eqref{eq: A11 equation 1} then leads to $\iden \geq \tilde{\mfa}_{11}\tilde{\mfa}_{11}^{\dagger}\geq0$.
	However, $|\det(\tilde{\mfa}_{11}\tilde{\mfa}_{11}^{\dagger})| = 1$, forcing the equality  $\iden = \tilde{\mfa}_{11}\tilde{\mfa}_{11}^{\dagger}$.
	In other words, $\tilde{\mfa}_{11}$ is a unitary matrix and $\tilde{\mfb}_{11} = 0$, $\tilde{\mfa}_{12} = 0$. 
	Eqs.~\eqref{eq: A12 equation 1} and~\eqref{eq: A21 equation 1} simplifies to $\tilde{\mfb}_{12} = 0$ and $\tilde{\mfb}_{21} = 0$.
	Combining everything, we arrive at the block diagonal structure
	\begin{align}
		\tilde{\mfA} = \tilde{\mfa}_{11} \oplus \tilde{\mfa}_{22},\quad \tilde{\mfB} = 0 \oplus \tilde{\mfb}_{22},\quad \mfD = \mfd \oplus 0.
	\end{align}
	
	Finally, we define $\tilde{A} \coloneqq WAW^{\dagger}$, $\tilde{B} \coloneqq WBW^{\dagger}$, $\tilde{\chi} \coloneqq W^{*}\chi W^{\dagger}$, and $\tilde{\mu} \coloneqq W\mu W^{\dagger}$.
	Note that $B\geq \pm\frac{1}{2}(\iden - AA^{\dagger})$ is equivalent to $\tilde{B} \geq \pm\frac{1}{2}(\iden - \tilde{A}\tilde{A}^{\dagger})$, i.e. the pair $(\tilde{A},\tilde{B})$ represents a valid $m$-mode GCO with a fixed-point $(\tilde{\mu},\tilde{\chi})$.
	Moreover $(\tilde{\mu},\tilde{\chi})$ represents the second moment that can be achieved by applying the GCO unitary $(W,0)$ to the original state with $(\mu,\chi)$.
	Using this notation, Eq.~\eqref{eq: mf var definitions} can be written as
	\begin{align}
		\tilde{A} = \tilde{\mu}^{\frac{1}{2}} \tilde{\mfA} \tilde{\mu}^{-\frac{1}{2}},\quad
		\tilde{B} = \tilde{\mu}^{\frac{1}{2}} \tilde{\mfB} \tilde{\mu}^{\frac{1}{2}},\quad
		\tilde{\chi} = (\tilde{\mu}^{*})^{\frac{1}{2}} \mfD \tilde{\mu}^{\frac{1}{2}}.
	\end{align}
	Let $\{\ket{i}\}_{i=1}^{m'}$ be an orthonormal basis of the support of $\mfD$ and $\{\ket{i}\}_{m'+1}^{m}$ be an orthonormal basis of the kernel of $\mfD$.
	This also implies that $\bra{i}\tilde{\mfB}\ket{i} = 0$ for all $i = 1,\cdots,m'$.
	
	Since $\tilde{\mu}>0$, the set $\{\tilde{\mu}^{-\frac{1}{2}}\ket{i}\}_{i=1}^{m}$ is a collection of $m$ linearly independent vectors. 
	Then there exists a unitary operator $V^{\dagger}$ that maps vectors in $\mathrm{span}\{\ket{i}\}_{i=1}^{m'}$ to those in $\mathcal{V} \coloneqq \mathrm{span}\{\tilde{\mu}^{-\frac{1}{2}}\ket{i}\}_{i=1}^{m'}$ and vectors in $\mathrm{span}\{\ket{i}\}_{i=m'+1}^{m}$ to those in $\mathcal{W} \coloneqq \mathcal{V}^{\perp}$.
    Observe that 
	\begin{align}
		\tilde{B}V^{\dagger}\ket{i} = 0,\quad \bra{i}V\tilde{B} = 0,\quad &\forall i=1,\cdots,m',\label{eq: B tilde ker}\\
		V^{*}\tilde{\chi}V^{\dagger}\ket{i} \neq 0,\quad &\forall i=1,\cdots,m',
	\end{align}
	which implies the block diagonal form
	\begin{align}
		\hat{B} \coloneqq V\tilde{B}V^{\dagger} = VWBW^{\dagger}V^{\dagger} = 0 \oplus \hat{b}_{22}.\label{eq: B hat def}
	\end{align}
	Furthermore, note that $\mathcal{V} = \mathrm{span}\{V^{\dagger}\ket{i}\}_{i=1}^{m'}$ is an $m'$-dimensional subspace of $\supp(\tilde{\chi})$.
	However, the dimensions of the support and kernel spaces of $\tilde{\chi}$ are the same as those of $\mfD$, because $\dim\mathrm{span}\{\ket{i}\}_{i=m'+1}^{m} = \dim\mathrm{span}\{\tilde{\mu}^{-\frac{1}{2}}\ket{i}\}_{i=m'+1}^{m}$.
	This means that $\mathcal{V} = \supp(\tilde{\chi})$ and $\mathcal{W} = \ker(\tilde{\chi})$, and that 
	\begin{align}
		\tilde{\chi}V^{\dagger}\ket{i} = 0,\quad \bra{i}V^{*}\tilde{\chi} = 0,\quad \forall i = m'+1,\cdots,m.\label{eq: chi tilde ker}
	\end{align}
	Eq.~\eqref{eq: chi tilde ker} again indicates the block diagonal form 
	\begin{align}
		\hat{\chi} \coloneqq V^{*}\tilde{\chi}V^{\dagger} = V^{*}W^{*}\chi W^{\dagger}V^{\dagger} = \hat{\chi}_{11} \oplus 0,\label{eq: chi hat def}
	\end{align}
	where $\hat{\chi}_{11}$ is a full rank $m'\times m'$ symmetric matrix. 
	It is also possible to choose $V$ so that $\hat{\chi}_{11}>0$ by $V \mapsto V (V'_{11}\oplus\1)$ with some $m'\times m'$ unitary $V'_{11}$ that does not change the block structure of other matrices.
	
	Now we focus on the matrix $A$. 
	Eq.~\eqref{eq: B hat def} and the inequality $\hat{B} \geq \pm\frac{1}{2}(\iden - \hat{A}\hat{A}^{\dagger})$ implies 
	\begin{align}
		\hat{A} \coloneqq VWAW^{\dagger}V^{\dagger} = \begin{pNiceArray}{c:c}
			\iden & 0 \\
			\hdottedline 
			0 & \hat{A}_{22}
		\end{pNiceArray}\breve{V} = \begin{pNiceArray}{c:c}
		\breve{v}_{11} & \breve{v}_{12}\\
		\hdottedline
		\hat{A}_{22}\breve{v}_{21} & \hat{A}_{22}\breve{v}_{22}
		\end{pNiceArray},\label{eq: hat A blocks}
	\end{align}
	for some unitary $\breve{V}$ with blocks $\breve{v}_{ij}$ and an $(m-m')\times (m-m')$-matrix $\hat{A}_{22}$.
	Eq.~\eqref{eq: chi invariant} can then be written as $\hat{\chi} = \hat{A}^{*}\hat{\chi}\hat{A}^{\dagger}$, or equivalently
	\begin{align}
		\begin{pNiceArray}{c:c}
			\hat{\chi}_{11} & 0\\
			\hdottedline
			0 & 0
		\end{pNiceArray} = \begin{pNiceArray}{c:c}
			\breve{v}_{11}^{*}\hat{\chi}_{11}\breve{v}_{11}^{\dagger} & \breve{v}_{11}^{*}\hat{\chi}_{11}\breve{v}_{21}^{\dagger}\hat{A}_{22}^{\dagger}\\
			\hdottedline
			\hat{A}_{22}^{*}\breve{v}_{21}^{*}\hat{\chi}_{11}\breve{v}_{11}^{\dagger} &  \hat{A}_{22}^{*}\breve{v}_{21}^{*}\hat{\chi}_{11}\breve{v}_{21}^{\dagger}\hat{A}_{22}^{\dagger}
		\end{pNiceArray}.\label{eq: chi hat invar}
	\end{align}
	Eq.~\eqref{eq: chi hat invar} and $\hat{\chi}_{11}>0$ lead to $\hat{A}_{22}\breve{v}_{21} = 0$, annulling the bottom left block in Eq.~\eqref{eq: hat A blocks}.
	Eq.~\eqref{eq: chi hat invar} then reduces to the equality $\hat{\chi}_{11} = \breve{v}_{11}^{*}\hat{\chi}_{11}\breve{v}_{11}^{\dagger}$.
	Since $\hat{\chi}_{11}$ is a full rank matrix, we have 
	\begin{align}
		|\det(\hat{\chi}_{11})| = |\det(\breve{v}_{11}^{*})\det(\hat{\chi}_{11})\det(\breve{v}_{11}^{\dagger})|.\label{eq: determinant hat chi}
	\end{align}
	On the other hand, the unitarity of $\breve{V}$ means that $0 \leq \breve{v}_{11}\breve{v}_{11}^{\dagger}  = \iden - \breve{v}_{12}\breve{v}_{12}^{\dagger} \leq \iden$. 
	If $\breve{v}_{11}\breve{v}_{11}^{\dagger} \neq \iden$, we have $|\det(\breve{v}_{11})|<1$, which contradicts Eq.~\eqref{eq: determinant hat chi}.
	Hence, $\breve{v}_{11}$ is a unitary matrix and $\breve{v}_{12} = 0$, annulling the top right block in Eq.~\eqref{eq: hat A blocks}.
	As a result, we obtain
	\begin{align}
		\hat{A} \coloneqq VWAW^{\dagger}V^{\dagger} =\hat{a}_{11}\oplus \hat{a}_{22},
	\end{align}
	where $\hat{a}_{11}$ is an $m'\times m'$ unitary matrix. 
	Setting $U = VW$, we conclude the proof.
\end{proof}

\section{Closure of gap between TOs and EnTOs in the Gaussian regime}\label{ap: EnTO}
In this section we prove one of our central results (Theorem \ref{th:gento_gto}, reiterated as Theorem \ref{th:gento_gto_app}), which is the fact that the gap between thermal operations and its enhanced version closes in the Gaussian regime. The proof relies largely on the full characterization of Gaussian covariant operations having their structure fully specified by the mathematical characterization of such operations, as outlined in Lemma \ref{le:free_dilation_app}.
\begin{theorem}[GEnTOs = GTOs]\label{th:gento_gto_app}
For any Gaussian enhanced thermal operation $\mE_{\mathrm{GEnTO}}$ on a system $S$, there exist an auxiliary Gibbs state $\gamma_R=\frac{e^{-\beta H_R}}{\tr(e^{-\beta H_R})}$ and a GCO unitary (i.e. a Gaussian unitary that is U(1)-covariant with respect to the Hamiltonian-generated representation) $\mathcal{U}_{\mathrm{PC}}$ acting on $SR$, such that
\begin{align}
    \mE_{\mathrm{GEnTO}}(\cdot)=\tr_R\left[\mathcal{U}_{\mathrm{PC}}(\cdot\otimes\gamma_R)\mathcal{U}_{\mathrm{PC}}^\dagger\right].
\end{align}
This implies GEnTOs and Gaussian thermal operations (GTOs) are equivalent on the level of channels.
\end{theorem}
\begin{proof}
The second moment of the Gibbs state reads $(\mu_\gamma^\beta,0)$, where $\mu_\gamma^\beta=(\bar{n}_\beta+\frac{1}{2})\iden$ and $\bar{n}_\beta=\frac{1}{e^{\beta\omega}-1}$.
The Gibbs-preserving condition is then written as
\begin{align}
    (\bar{n}_{\beta}+\frac{1}{2})AA^\dagger +B= (\bar{n}_{\beta}+\frac{1}{2})\iden.\label{eq:fix_point}
\end{align}
Condition (F1) in Lemma~\ref{le:free_dilation_app} is ensured by $B\geq0$. For condition (F2), we rewrite Eq. (\ref{eq:fix_point}) as
\begin{align}
    B=(\bar{n}_{\beta}+\frac{1}{2})(\iden-AA^\dagger).
\end{align}
It follows directly that $B_+$ and $\Lambda_+$ in Eqs. (\ref{eq:I_AA_SVD}) and (\ref{eq:B_decom}) are related as
\begin{align}
    B_+=(\bar{n}_{\beta}+\frac{1}{2})\Lambda_+.
\end{align}
Therefore, the dilation for a GCO with fixed-point $(\mu_\gamma^\beta,0)$ is $\mu_R=\mu_\gamma^\beta$ and $V$ in the form of Eq. (\ref{eq:dila_V}). In other words, any Gaussian EnTO $\mE$ can be decomposed as $\mE=\mU_L\circ\mE_D\circ\mU_R$, where $\mU_L$ and $\mU_R$ are Gaussian phase-covariant unitaries, and $\mE_D$ is realized by mixing the $j$th system mode ($j=1,\dots,m_+$) with an auxiliary mode initially in the Gibbs state at inverse temperature $\beta$ via a beam-splitter with transmittance $\eta_j=1-\lambda_j$. It means that the sets of GEnTOs and GTOs are equivalent. 
\end{proof}

\section{The monotones of Gaussian covariant channels}
\label{ap:proofs2}

\subsection{Proof of finiteness and faithfulness (P1 in main text)}
Recall the definition of monotones 
\begin{align}\label{eq: Slpm def in main}
	\Sl_{\pm}(\mu,\chi) \coloneqq \boldsymbol{\sigma}_{1}\left[ \left(\mu^{*}\pm\frac{\iden}{2}\right)^{-\frac{1}{2}}\chi\ \left(\mu\pm\frac{\iden}{2}\right)^{-\frac{1}{2}}\right],
\end{align}
where $(\cdot)^{-\frac{1}{2}}$ denotes the pseudoinverse of $(\cdot)^{\frac{1}{2}}$ and $\boldsymbol{\sigma}_{1}[\cdot]$ denotes the largest singular value.
By definitions of singular values and pseudoinverse, $0\leq \Sl_{\pm}(\mu,\chi) < \infty$.

To prove the faithfulness, we need to show that $\ker(\mu\pm\frac{\iden}{2})\subseteq \ker(\chi)$.
\begin{proposition}\label{proposition: kernels inclusion}
	$\ker(\mu\pm\frac{\iden}{2})\subseteq \ker(\chi)$ for any second moments of a quantum state $(\mu,\chi)$.
\end{proposition}
\begin{proof}
	First, $\ker(\mu+\frac{\iden}{2})\subseteq \ker(\chi)$ is obvious from $\mu+\frac{\iden}{2}>0$.
For the second part, recall that 
\begin{align}
	M+\mmZ = \begin{pmatrix}
		\mu^{*} + \frac{\iden}{2} & \chi\\
		\chi^{*} & \mu-\frac{\iden}{2}
	\end{pmatrix} \geq0,
\end{align}
from the uncertainty principle for the second moments.
For any $m$-dimensional vector $\psi^{\dagger}$, a direct calculation using the above equation shows that
\begin{align}
	\begin{pmatrix}
		0 & \psi
	\end{pmatrix} (M+\mmZ) \begin{pmatrix}
		0\\ \psi^{\dagger}
	\end{pmatrix} = \psi(\mu-\frac{\iden}{2})\psi^{\dagger},
\end{align}
and since both $\mu-\frac{\iden}{2}$ and $M+\mmZ$ are positive semidefinite,
\begin{align}\label{eq: equivalence of mu and chi kers}
	\psi(\mu-\frac{\iden}{2})\psi^{\dagger} = 0 \quad \iff \quad (M+\mmZ)\begin{pmatrix}
		0 \\ \psi^{\dagger}
	\end{pmatrix} = \begin{pmatrix}
		\chi\psi^{\dagger} \\ (\mu-\frac{\iden}{2})\psi^{\dagger}
	\end{pmatrix} = 0,
\end{align}
proving $\ker(\mu-\frac{\iden}{2})\subseteq \ker(\chi)$.
\end{proof}

On a related note, we show that without loss of generality, we can focus on states characterized by the second moments $(\mu,\chi)$ that fulfill the condition $\mu - \iden /2 > 0$, in other words, we can safely discard the modes that correspond to zero eigenvalues. This is because we show that other states that do not satisfy this condition can nevertheless be reversibly transformed via GCOs to such states.

\begin{lemma}\label{le:vac-block-diag}
	Let $(\mu,\chi)$ represent the second moment of an $m$-mode state such that $\mu-\frac{\iden}{2}$ is not positive definite. 
	Then there exists a unitary GCO $(V,0)$ and an $m'$-mode state with $0<m'<m$ with the second moment $(\mu'_{S'},\chi'_{S'})$ such that $V\mu V^{\dagger} = \frac{\iden}{2} \oplus \mu'_{S'}$ and $V^{*}\chi V^{\dagger} = 0 \oplus \chi'_{S'}$, where
	\begin{itemize}
		\item $\mu'_{S'}-\frac{\iden}{2}$ is positive definite,
		\item $(\mu_{S},\chi_{S})\toGCO(\mu'_{S'},\chi'_{S'})$ and $(\mu'_{S'},\chi'_{S'}) \toGCO (\mu_{S},\chi_{S})$, i.e. the transformations between the two are reversible.
	\end{itemize}
\end{lemma}
\begin{proof}
	Let $V^{\dagger}$ be the unitary that transforms the kernel of $\mu-\frac{\iden}{2}$ into the first $(m-m')$ basis vectors in the usual basis. 
	Then the Hermitian matrix $V(\mu-\frac{\iden}{2})V^{\dagger} = 0 \oplus (\mu'_{S'} - \frac{\iden}{2})$ with some positive definite $m'\times m'$ matrix $(\mu'_{S'} - \frac{\iden}{2})>0$.
	From Proposition~\ref{proposition: kernels inclusion}, we have $\ker(\mu-\frac{\iden}{2})\subseteq\ker(\chi)$, which implies $V^{*}\chi V^{\dagger} = 0 \oplus \chi'_{C'}$ with some symmetric $m'\times m'$ matrix $\chi'_{C'}$.
	The transformation $(\mu_{S},\chi_{S})\toGCO(\mu'_{S'},\chi'_{S'})$ can be made by i) applying the GCO unitary $(V,0)$ and ii) tracing out the first $(m-m')$ modes. 
	Conversely, $(\mu'_{S'},\chi'_{S'}) \toGCO (\mu_{S},\chi_{S})$ is realized by i) appending $(m-m')$ vacuum modes and ii) applying the GCO unitary $(V^{\dagger},0)$.
\end{proof}

\subsection{Proof of monotonicity (P2 in main text)}

	For the rest of our proofs, we often use an alternative definition of $\Sl_{\pm}$ that is equivalent to Definition~\ref{def: Sl monotones} in the main text.
	\begin{lemma}\label{lemma: Slpm takagi}
		Let $(\mu,\chi)$ be the second moments of an $m$-mode quantum state with $\chi\neq0$. 
		Then
		\begin{align}\label{eq: Slpm def in app}
			\Sl_{\pm}(\mu,\chi)= \sup\limits_{\psi^*\chi\psi^\dagger\neq0}  \frac{|\psi^*\chi\psi^\dagger|}{\psi (\mu\pm\frac{\iden}{2})\psi^\dagger},
		\end{align}
		which are in the form of generalized Rayleigh quotients maximized over $m$-dimensional row vectors $\psi$.
	\end{lemma}
	\begin{proof}
		We start from Eq.~\eqref{eq: Slpm def in main} and derive Eq.~\eqref{eq: Slpm def in app}.
		Since the matrix $K_{\pm} \coloneqq (\mu^{*}\pm\frac{\iden}{2})^{-\frac{1}{2}}\chi\ (\mu\pm\frac{\iden}{2})^{-\frac{1}{2}}$ is symmetric, the largest singular value $\Sl_{\pm}(\mu,\chi)$ can be written as 
		\begin{align}\label{eq:Sl_pm_def_compact}
			\Sl_\pm(\mu,\chi)=\sup_{\|\phi\|=1}\left|\phi^{\tp} K_{\pm} \phi\right| = \sup_{\phi}\frac{\left |\phi^{\tp}K_{\pm}\phi\right |}{\phi^{\dagger}\phi}.
		\end{align}
		Note that the supremum is achieved only when $\phi \in \supp(\mu\pm\frac{\iden}{2})$ because the pseudoinverse operator $(\mu\pm\frac{\iden}{2})^{-\frac{1}{2}}$ inside $K_{\pm}$ annihilate any component of $\phi$ that is not in the support of $\mu\pm\frac{\iden}{2}$.
		Define $\psi^{\dagger} = (\mu\pm\frac{\iden}{2})^{-\frac{1}{2}}\phi$ and we obtain $|\phi^{\tp} K_{\pm} \phi| = |\psi^{*} \chi \psi^{\dagger}|$ for the numerator, and $\phi^{\dagger}\phi = \psi (\mu\pm\frac{\iden}{2})\psi^\dagger$ for the denominator using $\phi \in \supp(\mu\pm\frac{\iden}{2})$. 
		Now we can write 
		\begin{align}\label{eq: Slpm def with constraints}
			\Sl_{\pm}(\mu,\chi)= \sup\limits_{\psi^{\dagger}\in\supp(\mu\pm\frac{\iden}{2})}  \frac{|\psi^*\chi\psi^\dagger|}{\psi (\mu\pm\frac{\iden}{2})\psi^\dagger} = \sup\limits_{\substack{\psi^{\dagger}\in\supp(\mu\pm\frac{\iden}{2})\\ \psi^*\chi\psi^\dagger\neq0}}  \frac{|\psi^*\chi\psi^\dagger|}{\psi (\mu\pm\frac{\iden}{2})\psi^\dagger},
		\end{align}
		where the second equality follows from the fact that $\Sl_{\pm}(\mu,\chi)\geq0$ and $\psi^*\chi\psi^\dagger=0$ gives $\Sl_{\pm}(\mu,\chi) = 0$.
		The only difference between Eqs.~\eqref{eq: Slpm def in app} and~\eqref{eq: Slpm def with constraints} is the additional constraint $\psi^{\dagger}\in\supp(\mu\pm\frac{\iden}{2})$ for the latter. 
		
		A general vector $\tilde{\psi}^{\dagger}$ can be written as $\tilde{\psi}^{\dagger} = \psi^{\dagger} + \zeta^{\dagger}$ where $\psi^{\dagger}\in\supp(\mu\pm\frac{\iden}{2})$ and $\zeta^{\dagger}\in\ker(\mu\pm\frac{\iden}{2})$.
		We substitute $\psi^{\dagger}$ with this $\tilde{\psi}^{\dagger}$ in the RHS of Eq.~\eqref{eq: Slpm def in app}.
		Proposition~\ref{proposition: kernels inclusion} implies that $\chi\zeta^{\dagger} = 0$ and $\frac{|\tilde{\psi}^*\chi\tilde{\psi}^\dagger|}{\tilde{\psi} (\mu\pm\frac{\iden}{2})\tilde{\psi}^\dagger} = \frac{|\psi^*\chi\psi^\dagger|}{\psi (\mu\pm\frac{\iden}{2})\psi^\dagger}$.
		Therefore, the additional constraint $\psi^{\dagger}\in\supp(\mu\pm\frac{\iden}{2})$ does not change the supremum value of Eq.~\eqref{eq: Slpm def in app}, and we prove the desired equivalence. 
	\end{proof}

Now we prove the monotonicity of $\Sl_{\pm}$ using Lemma~\ref{lemma: Slpm takagi}.
 
\begin{lemma}[GCO monotones]\label{le:monotones_app}
	Whenever $(\mu,\chi)\toGCO(\mu',\chi')$ is possible for some $\Inm$-mode and $\Outm$-mode systems $(\mu,\chi)$ and $(\mu',\chi')$, then 
    \begin{align}
        \Sl_{\pm}(\mu,\chi)\geq\Sl_{\pm}(\mu',\chi').
    \end{align}
\end{lemma}
\begin{proof}
	Suppose that $(\mu,\chi)\toGCO(\mu',\chi')$. 
	From Lemma~\ref{le:gene_GCO_app} and the uncertainty principle, we have that the initial and final second moments are related to each other by 
	\begin{align}
		\mu'\pm\frac{\iden}{2 }& \geq  A(\mu\pm\frac{\iden}{2})A^\dagger \geq 0, \label{eq:mup_vs_mu}\\
		\chi' & =  A^*\chi A^\dagger.\label{eq:chip_vs_chi}
	\end{align}
	Suppose that $\bar\psi$ is the optimal vector for the form of $\Sl_{\pm}$ in Eq.~\eqref{eq: Slpm def in app} achieving 
	\begin{align}
		\Sl_{\pm}(\mu',\chi') = \frac{|\bar\psi^*\chi'\bar\psi^\dagger|}{\bar\psi (\mu'\pm\frac{\iden}{2})\bar\psi^\dagger}.
	\end{align}
	Let us consider the vector $\bar\psi_A\coloneqq\bar\psi A$. Due to Eq.~\eqref{eq:chip_vs_chi}, this is simply another $m$-dimensional vector that gives rise to $\bar\psi^{*}\chi'\bar\psi^{\dagger} = \bar\psi_{A}^{*}\chi\bar\psi_{A}^{\dagger}$.
	If $\bar\psi_A^* \chi\bar\psi_A^\dagger=0$, then $\Sl_\pm(\mu',\chi')=0$, and $\Sl_\pm(\mu,\chi)\geq\Sl_\pm(\mu',\chi')$ trivially holds from the non-negativity of $\Sl_\pm$. 
	Otherwise, $\bar\psi_A^* \chi\bar\psi_A^\dagger\neq0$, implying $\bar\psi_{A}(\mu-\frac{\iden}{2})\bar\psi_A^{\dagger} \neq 0$ from Proposition~\ref{proposition: kernels inclusion}. 
	In this case, from Eq.~\eqref{eq:mup_vs_mu} it follows that 
	\begin{align}\label{eq: monotonicity-easy}
		\Sl_{\pm}(\mu',\chi') \leq \frac{|\bar\psi_A^* \chi\bar\psi_A^\dagger|}{\bar\psi_A(\mu\pm\frac{\iden}{2})\bar\psi_A^\dagger}\nonumber\leq \Sl_{\pm}(\mu,\chi).
	\end{align}
    This concludes the proof.
\end{proof}

\subsection{Proof of complete non-extensiveness (P3 in main text)}

We establish a key property of the above GCO monotone, which will be useful to establish monotonicity under catalytic operations. In particular, we note the general non-extensiveness of the GCO monotone.
\begin{lemma}[Lower bounds on the GCO monotones]\label{le:monotones-composite}
	Let $(\mu,\chi)$ be the second moments of a $(m_{S}+m_{C})$-dimensional system whose block matrix form is 
	\begin{align}
		\mu = \begin{pNiceArray}{c:c}
			\mu_{S} & \mu_{SC} \\
			\hdottedline
			\mu_{CS} & \mu_{C}
		\end{pNiceArray},\quad 
		\chi = \begin{pNiceArray}{c:c}
			\chi_{S} & \chi_{SC} \\
			\hdottedline
			\chi_{CS} & \chi_{C}
		\end{pNiceArray},
	\end{align}
	 with $m_{S}\times m_{S}$-blocks $\mu_{S},\chi_{S}$ and $(m_{C}\times m_{C})$-blocks $\mu_{C},\chi_{C}$.
	 Then
     \begin{align}\label{eq:lowerbound}
         \Sl_{\pm}(\mu,\chi) \geq \max\{\Sl_{\pm}(\mu_{S},\chi_{S}), \Sl_{\pm}(\mu_{C},\chi_{C})\}.
     \end{align}
	 Furthermore, equality for Eq.~\eqref{eq:lowerbound} is achieved for systems that are uncorrelated, i.e. $\mu_{SC} = \chi_{SC} = 0$. 
\end{lemma}
\begin{proof}
	The first part of the lemma statement is easy to prove using Lemma~\ref{lemma: Slpm takagi}.
	It follows from the fact that trivial extensions of $\psi$ that enter the optimization of $\Sl_{\pm}(\mu,\chi)$ in Eq.~\eqref{eq: Slpm def in app} for individual systems are already included in the optimization across the global system. 
	More concretely, let us start by assuming that $\mu_{S} \neq \frac{\iden}{2}$ and $\mu_{C} \neq \frac{\iden}{2}$.
	Let us additionally consider the row vectors $\psi_{S}\in\mathbb{C}^{m_{S}}$, $\psi_{C}\in\mathbb{C}^{m_{C}}$ that produce
	\begin{align}
		\Sl_{\pm}(\mu_{S},\chi_{S}) &= \frac{|\psi_{S}^{*}\chi_{S}\psi_{S}^{\dagger}|}{\psi_{S}(\mu_{S}\pm\frac{\iden}{2})\psi_{S}^{\dagger}}, \qquad
		\Sl_{\pm}(\mu_{C},\chi_{C}) = \frac{|\psi_{C}^{*}\chi_{C}\psi_{C}^{\dagger}|}{\psi_{C}(\mu_{C}\pm\frac{\iden}{2})\psi_{C}^{\dagger}}.
	\end{align}
	If we choose $\psi = (\psi_{S},0)\in\mathbb{C}^{m_{S}+m_{C}}$, 
	\begin{align}
		\Sl_{\pm}(\mu_{S},\chi_{S})=  \frac{|\psi^{*}\chi\psi^{\dagger}|}{\psi(\mu\pm\frac{\iden}{2})\psi^{\dagger}} \leq \Sl_{\pm}(\mu,\chi).
	\end{align}
	Similarly, $\Sl_{\pm}(\mu_{C},\chi_{C}) \leq \Sl_{\pm}(\mu,\chi)$ can be shown by choosing $\psi = (0,\psi_{C})\in\mathbb{C}^{m_{S}+m_{C}}$.
	Next, note that whenever $\mu_{S} = \frac{\iden}{2}$ and/or $\mu_{C} = \frac{\iden}{2}$, then $\Sl_{\pm}(\mu_{S},\chi_{S}) = 0$ and/or $\Sl_{\pm}(\mu_{C},\chi_{C}) = 0$, giving the same conclusion. 
	
	The special case of equality can also be easily shown using the original definition of $\Sl_{\pm}$.
	Denote $K_{\pm}(\mu,\chi) = (\mu^{*}\pm\frac{\iden}{2})^{-\frac{1}{2}}\chi\ (\mu\pm\frac{\iden}{2})^{-\frac{1}{2}}$; then $\Sl_{\pm}(\mu,\chi)$ is the largest singular value of $K_{\pm}(\mu,\chi)$. 
	It is straightforward to check that $K_{\pm}(\mu_{S}\oplus\mu_{C},\chi_{S}\oplus\chi_{C})=K_{\pm}(\mu_{S},\chi_{S})\oplus K_{\pm}(\mu_{C},\chi_{C})$. 
	Therefore, the largest singular value of the LHS is the larger one of the largest singular values of $K_{\pm}(\mu_{S},\chi_{S})$ and $K_{\pm}(\mu_{C},\chi_{C})$.
\end{proof}

\subsection{Monotonicity under correlated catalysis (P4 in main text)}
The next two lemmas address how our GCO monotones behave under correlated catalytic transformations. Firstly, we again prove that without loss of generality, the catalyst state can be characterized by some $\tilde{\mu}_{C'}>\frac{\iden}{2}$.

\begin{lemma}\label{le: catalyst-vac-removal}
	Suppose that a GCO process $(\mu_{S},\chi_{S}) \to (\mu'_{S},\chi'_{S})$ can be achieved by a catalyst $(\mu_{C},\chi_{C})$.
	Then this process can also be achieved by a catalyst $(\tilde{\mu}_{C'},\tilde{\chi}_{C'})$, such that $\tilde{\mu}_{C'}>\frac{\iden}{2}$.
\end{lemma}
\begin{proof}
	Suppose that $\mu_{C}-\frac{\iden}{2}$ is not positive definite. 
	Then Lemma~\ref{le:vac-block-diag} applies, and we can write
	\begin{align}
		V\mu_{C}V^{\dagger} = \frac{\iden}{2} \oplus \tilde{\mu}_{C'},\quad V^{*}\chi_{C} V^{\dagger} = 0 \oplus \tilde{\chi}_{C'}, 
	\end{align} 
	with a unitary $V$, for some $(\tilde{\mu}_{C'}, \tilde{\chi}_{C'})$ with $\tilde{\mu}_{C'}>\frac{\iden}{2}$.
	Denote the original GCO applied to $(\mu_{S}\oplus\mu_{C},\chi_{S}\oplus\chi_{C})$ as $(A,B)$.
	We can construct the following catalytic process starting from $(\mu_{S}\oplus\tilde{\mu}_{C'},\chi_{S}\oplus\tilde{\chi}_{C'})$:
	\begin{enumerate}
		\item Append some vacuum modes to obtain $(\mu_{S}\oplus\frac{\iden}{2} \oplus \tilde{\mu}_{C'},\chi_{S}\oplus0 \oplus \tilde{\chi}_{C'})$.
		\item Apply the unitary GCO $(\iden\oplus V^{\dagger},0)$ to obtain $(\mu_{S}\oplus\mu_{C},\chi_{S}\oplus\chi_{C})$.
		\item Apply the GCO $(A,B)$ to obtain the desired system marginal moments $(\mu'_{S},\chi'_{S})$ while keeping the catalyst marginal moments $(\mu_{C},\chi_{C})$.
		\item Apply the unitary GCO $(0\oplus V,0)$ to transform catalyst moments to be $(\frac{\iden}{2} \oplus \tilde{\mu}_{C'},0 \oplus \tilde{\chi}_{C'})$.
		\item Remove vacuum modes and retrieve the catalyst $(\tilde{\mu}_{C'},\tilde{\chi}_{C'})$.
	\end{enumerate}
	Since all five steps are valid GCOs, the lemma statement holds true.
\end{proof}

\begin{lemma}\label{le: no catalysis} 
	$\Sl_{\pm}(\mu,\chi)\geq\Sl_{\pm}(\mu',\chi')$ whenever $(\mu,\chi)$ can be transformed into $(\mu',\chi')$ with the help of some catalyst state $(\mu_{C},\chi_{C})$ that can be correlated at the end.
\end{lemma}

The initial states of $S$ and $C$ are labeled as $(\mu_S,\chi_S)$ and $(\mu_C,\chi_C)$ respectively. After the action of a GCO $(A,B)$, the reduced states becomes $(\mu'_S,\chi'_S)$ and $(\mu'_C,\chi'_C)$. 
In the following, we will sketch the proof of $\Sl_\pm(\mu_S',\chi_S')\leq \Sl_\pm(\mu_S,\chi_S)$ from the catalytic condition $(\mu_C,\chi_C)=(\mu'_C,\chi'_C)$.

Lemma~\ref{lemma: Slpm takagi} singles out vectors $\psi$ achieving the optimal value for each $\Sl_\pm(\mu_C',\chi_C')$.
This vector can be regarded as the most asymmetric mode $C(\psi)$ among catalyst modes $C$ in some basis. 

\textbf{Case 2}: The second moment of this most asymmetric mode $\mu'_{C(\psi)}$ after the catalytic transformation depends on the input of the system modes $\mu_S$. 
We prove that the catalytic condition implies $\Sl_\pm(\mu_C,\chi_C)\leq\Sl_\pm(\mu_S,\chi_S)$. 
Using Lemma~\ref{le:monotones_app} and~\ref{le:monotones-composite}, we obtain $\Sl_\pm(\mu'_S,\chi'_S)\leq \Sl_\pm(\mu_S,\chi_S)$.

\textbf{Case 3-1}: $\mu_C'$ does not depend on $\mu_S$, i.e., the second moment of all the catalytic modes is not affected by the input of the system modes. 
Then the catalytic condition immediately implies that the state of the catalyst is preserved by a smaller GCO. 
By Lemma~\ref{le: fixed point decomposition}, the asymmetric modes in $C$ should evolve independently. Effectively, the system modes can only interact with catalyst modes with vanishing $\chi$, and it follows directly that $\Sl_\pm$ of $S$ is not increased.

\textbf{Case 3-2}: $\mu_C'$ depends on $\mu_S$ but $\mu'_{C(\psi)}$ does not depend on $\mu_S$, i.e., although the second moment of the entire catalytic modes is affected by the input state of the system modes, the most asymmetric one is unaffected.
This can happen only if the number of modes in $C$ is no less than 2. 
For this case, we prove that at least one of the most asymmetric catalytic modes should evolve independently, and $S$ can only interact with the remaining catalytic modes labeled as $C^{(1)}$.

With this new catalyst $C^{(1)}$, we can again identify the most asymmetric mode $C^{(1)}(\psi)$ and iterate the above discussion, which will either prove the Lemma or yield a new catalyst $C^{(2)}$ and so on.
In each iteration, the size of $C^{(j)}$ is strictly reduced, and the iteration terminates as long as it does not fall into \textbf{Case 3-2}. 
In the worst case, we arrive at the single mode catalyst $C^{(m_{C}-1)}$, which necessarily falls into either \textbf{Case 2} or \textbf{3-1} and ends the iteration. 
Below is the formal proof of this iteration.

\begin{proof}
	Suppose that $(\mu,\chi)$ is the second moment for the initial $SC$ state and $(\mu',\chi')$ the one for the final $SC$ state. 
	Using Lemma~\ref{le:monotones_app} and~\ref{le:monotones-composite}, we obtain
	\begin{align}\label{eq: initial S C monotones bounding final S monotone}
		\Sl_{\pm}(\mu,\chi) = \max\{\Sl_{\pm}(\mu_{S},\chi_{S}), \Sl_{\pm}(\mu_{C},\chi_{C})\} \geq \Sl_{\pm}(\mu',\chi') \geq \max\{\Sl_{\pm}(\mu'_{S},\chi'_{S}), \Sl_{\pm}(\mu'_{C},\chi'_{C})\}.
	\end{align}
	The first equality is true since it is always assumed that the initial state of $S$ and $C$ are uncorrelated. 
	Eq.~\eqref{eq: initial S C monotones bounding final S monotone} then implies
	\begin{align}
		\max\{\Sl_{\pm}(\mu_{S},\chi_{S}), \Sl_{\pm}(\mu_{C},\chi_{C})\} \geq \Sl_{\pm}(\mu'_{S},\chi'_{S}).
	\end{align}
	Hence, showing
	\begin{align}\label{eq:desired-inequality}
		\Sl_{\pm}(\mu_{C},\chi_{C}) \leq \Sl_{\pm}(\mu_{S},\chi_{S})
	\end{align}
	is sufficient to prove the monotonicity of $\Sl_{\pm}$ even under catalytic transformations; we do so in \textbf{Cases 1} and \textbf{2} below.
	
	We begin by explicitly writing the GCO $(A,B)$ acting on $SC$ as
	\begin{align}
		A=\begin{pNiceArray}{c:c}
			a_{SS} & a_{SC} \\
			\hdottedline
			a_{CS} & a_{CC}
		\end{pNiceArray},B=\begin{pNiceArray}{c:c}
			b_{SS} & b_{SC} \\
			\hdottedline
			b_{CS} & b_{CC}
		\end{pNiceArray},
	\end{align}
	where $a_{SS}$ and $b_{SS}$ are $m_{S}\times m_{S}$ matrices and $a_{CC}$ and $b_{CC}$ are $m_C\times m_C$ matrices. 
	Recall that if a matrix is positive semidefinite, then all its principal submatrices have to also be positive semidefinite. This means that $B\geq\pm\frac{1}{2}(\iden-AA^\dagger)$ implies
	\begin{align}
		b_{SS} & \geq \pm\frac{1}{2}(\iden-a_{SS}a_{SS}^{\dagger}-a_{SC}a_{SC}^\dagger),\label{eq:posi_b_11}\\
		b_{CC} & \geq \pm\frac{1}{2}(\iden-a_{CS}a_{CS}^\dagger-a_{CC}a_{CC}^\dagger).\label{eq:posi_b_22}
	\end{align}
	The marginal state of $S$ in the output is then calculated as
	\begin{align}
		\mu'_S & =a_{SS}\mu_{S}a_{SS}^{\dagger}+a_{SC}\mu_{C} a_{SC}^\dagger+b_{SS},\label{eq:mu_S_p}\\
		\chi_S' & = a^{*}_{SS}\chi_{S}a_{SS}^{\dagger}+a_{SC}^*\chi_{C} a_{SC}^\dagger.\label{eq:chi_S_p}
	\end{align}
	For $C$, the catalytic condition reads
	\begin{align}
		\mu_C & = a_{CS}\mu_S a_{CS}^\dagger+a_{CC}\mu_C a_{CC}^\dagger+b_{CC},\label{eq:cat-condition-1}\\
		\chi_C & =  a^{*}_{CS}\chi_S a_{CS}^\dagger+a_{CC}^{*}\chi_C a_{CC}^\dagger,\label{eq:cat-condition-2}
	\end{align}
	which can be rewritten as
	\begin{align}
		\mu_C\pm\frac{\iden}{2} & = a_{CS}(\mu_S \pm \frac{\iden}{2})a_{CS}^\dagger+a_{CC}(\mu_C \pm \frac{\iden}{2})a_{CC}^\dagger+\Delta_{\pm},\label{eq:cat-condition-1-alt}\\
		\chi_C & = a^{*}_{CS}\chi_S a_{CS}^\dagger+a_{CC}^{*}\chi_C a_{CC}^\dagger.\label{eq:cat-condition-2-alt}
	\end{align}
	where $\Delta_\pm\equiv b_{CC}\pm\frac{1}{2}(\iden-a_{CS}a_{CS}^\dagger-a_{CC}a_{CC}^\dagger)\geq0$ from Eq.~\eqref{eq:posi_b_22}.
	
	Let $\psi$ be a $m_{C}$-dimensional row vector achieving 
	\begin{align}
		\Sl_{\pm}(\mu_{C},\chi_{C}) = \frac{|\psi^{*}\chi_{C}\psi^{\dagger}|}{\psi (\mu_{C}\pm\frac{\iden}{2})\psi^{\dagger}}.
	\end{align}
	Our proof strategy is to replace $\chi_{C}$ and $\mu_{C}$ with the fixed-point conditions Eqs.~\eqref{eq:cat-condition-1} and~\eqref{eq:cat-condition-2}, which gives
	\begin{align}
		|\psi^{*}\chi_{C}\psi^{\dagger}| = |\psi^{*}a^{*}_{CS}\chi_{S} a_{CS}^{\dagger}\psi^{\dagger} +\psi^{*}a_{CC}^{*}\chi_{C} a_{CC}^{\dagger}\psi^{\dagger}|
	\end{align}
	for the numerator and 
	\begin{align}
		\psi (\mu_{C}\pm\frac{\iden}{2})\psi^{\dagger} = \psi a_{CS}(\mu_S \pm \frac{\iden}{2})a_{CS}^{\dagger}\psi^{\dagger} +\psi a_{CC}(\mu_C \pm \frac{\iden}{2})a_{CC}^{\dagger}\psi^{\dagger} + \psi\Delta_{\pm}\psi^{\dagger}
	\end{align}
	for the denominator. 
	To declutter the notation, we define an $m_{S}$-dimensional row vector $\psi_{CS} \coloneqq \psi a_{CS}$ and an $m_{C}$-dimensional row vector $\psi_{CC} \coloneqq \psi a_{CC}$ to write 
	\begin{align}
		|\psi^{*}\chi_{C}\psi^{\dagger}| &= |\psi_{CS}^{*}\chi_{S} \psi_{CS}^{\dagger} +\psi_{CC}^{*}\chi_{C} \psi_{CC}^{\dagger}|,\\
		\psi (\mu_{C}\pm\frac{\iden}{2})\psi^{\dagger} &= \psi_{CS} (\mu_S \pm \frac{\iden}{2})\psi_{CS}^{\dagger} +\psi_{CC}(\mu_C \pm \frac{\iden}{2})\psi_{CC}^{\dagger} + \psi\Delta_{\pm}\psi^{\dagger}.
	\end{align}
	
	\textbf{Case 1}: consider the trivial case of $\psi_{CS}(\mu_{S}\pm\frac{\iden}{2})\psi_{CS}^{\dagger}=0$ and $\psi_{CC}(\mu_{C}\pm\frac{\iden}{2})\psi_{CC}^{\dagger}=0$.
	Then, $\psi_{CS}^{*}\chi_{S}\psi_{CS}^{\dagger} = 0$ and $\psi_{CC}^{*}\chi_{C}\psi_{CC}^{\dagger} = 0$ from Proposition~\ref{proposition: kernels inclusion} implying that $\psi^*\chi_{C}\psi^\dagger = 0$ and $\Sl_{\pm}(\mu_{C},\chi_{C}) =0$.
	The non-negativity of $\Sl_{\pm}$ yields the desired inequality Eq.~\eqref{eq:desired-inequality}.\\
	
	\textbf{Case 2}: assume that $\psi_{CS}(\mu_{S}\pm\frac{\iden}{2})\psi_{CS}^{\dagger}\neq0$. Note that Eq.~\eqref{eq:cat-condition-1-alt} gives 
	\begin{align}
		\Sl_{\pm}(\mu_{C},\chi_{C}) 
        &\leq \frac{|\psi_{CS}^{*}\chi_{S}\psi_{CS}^{\dagger} + \psi_{CC}^{*}\chi_{C}\psi_{CC}^{\dagger}|}{\psi_{CS}(\mu_{S}\pm\frac{\iden}{2})\psi_{CS}^{\dagger} + \psi_{CC}(\mu_{C}\pm\frac{\iden}{2})\psi_{CC}^{\dagger}} \nonumber\\
        &\leq \frac{|\psi_{CS}^{*}\chi_{S}\psi_{CS}^{\dagger}| + |\psi_{CC}^{*}\chi_{C}\psi_{CC}^{\dagger}|}{\psi_{CS}(\mu_{S}\pm\frac{\iden}{2})\psi_{CS}^{\dagger} + \psi_{CC}(\mu_{C}\pm\frac{\iden}{2})\psi_{CC}^{\dagger}}\nonumber\\
		&\leq \frac{\Sl_{\pm}(\mu_{S},\chi_{S})\psi_{CS}(\mu_{S}\pm\frac{\iden}{2})\psi_{CS}^{\dagger} + \Sl_{\pm}(\mu_{C},\chi_{C})\psi_{CC}(\mu_{C}\pm\frac{\iden}{2})\psi_{CC}^{\dagger}}{\psi_{CS}(\mu_{S}\pm\frac{\iden}{2})\psi_{CS}^{\dagger} + \psi_{CC}(\mu_{C}\pm\frac{\iden}{2})\psi_{CC}^{\dagger}},\label{eq:psi-nonzero-case}
	\end{align}
    where the first inequality makes use of $\Delta_{\pm}\geq0$, the second simply being triangle inequality, and the last coming from the definition of $\Sl_\pm$ itself.
	Note that we used Proposition~\ref{proposition: kernels inclusion} again to take into account the case $\psi_{CC}(\mu_{C}\pm\frac{\iden}{2})\psi_{CC}^{\dagger}=0$.
	Eq.~\eqref{eq:psi-nonzero-case} states that $\Sl_{\pm}(\mu_{C},\chi_{C})$ is less than or equal to the convex combination of $\Sl_{\pm}(\mu_{S},\chi_{S})$ and itself, with the coefficients   
	\begin{align}
		\frac{\psi_{CS}(\mu_{S}\pm\frac{\iden}{2})\psi_{CS}^{\dagger}}{\psi_{CS}(\mu_{S}\pm\frac{\iden}{2})\psi_{CS}^{\dagger} + \psi_{CC}(\mu_{C}\pm\frac{\iden}{2})\psi_{CC}^{\dagger}} \in (0,1),\quad \frac{\psi_{CC}(\mu_{C}\pm\frac{\iden}{2})\psi_{CC}^{\dagger}}{\psi_{CS}(\mu_{S}\pm\frac{\iden}{2})\psi_{CS}^{\dagger} + \psi_{CC}(\mu_{C}\pm\frac{\iden}{2})\psi_{CC}^{\dagger}} \in (0,1).
	\end{align}
	Hence, Eq.~\eqref{eq:psi-nonzero-case} can be true only when Eq.~\eqref{eq:desired-inequality} is true. 
	
	Finally, assume that $\psi_{CS}(\mu_{S}\pm\frac{\iden}{2})\psi_{CS}^{\dagger}=0$
	but $\psi_{CC}(\mu_{C}\pm\frac{\iden}{2})\psi_{CC}^{\dagger}\neq0$. We need to analyze two separate cases in this situation.\\
	
	\textbf{Case 3-1}: if $a_{CS}(\mu_S \pm \frac{\iden}{2})a_{CS}^\dagger=0$, the catalyst $(\mu_{C},\chi_{C})$ is preserved by a GCO $(a_{CC},b_{CC})$
	\footnote{In particular, when $a_{CS} = 0$, the GCO $(A,B)$ represents a catalytic channel~\cite{Son2024RC} as the catalyst is preserved regardless of the system state. 
		We prove the impossibility of such a catalytic channel.}.
	Lemma~\ref{le: fixed point decomposition} implies that there exists a unitary matrix $U = \1_{S} \oplus U_{C}$ such that 
	\begin{align}
		UAU^{\dagger} = \begin{pNiceArray}{c:c:c}
			a_{SS} & a_{SC_{1}} & a_{SC_{2}}\\
			\hdottedline
			a_{C_{1}S} & a_{C_{1}C_{1}} & 0\\
			\hdottedline
			a_{C_{2}S} & 0 & a_{C_{2}C_{2}}
		\end{pNiceArray},\quad 
		UBU^{\dagger} = \begin{pNiceArray}{c:c:c}
			b_{SS} & b_{SC_{1}} & b_{SC_{2}}\\
			\hdottedline
			b_{C_{1}S} & 0 & 0\\
			\hdottedline
			b_{C_{2}S} & 0 & b_{C_{2}C_{2}}
		\end{pNiceArray},\quad 
		U^{*}\chi U^{\dagger} = \begin{pNiceArray}{c:c:c}
			\chi_{S} & 0 & 0\\
			\hdottedline
			0 & \chi_{C_{1}} & 0\\
			\hdottedline
			0 & 0 & 0
		\end{pNiceArray}
	\end{align}
	with a unitary matrix $a_{C_{1}C_{1}}$.
	The condition $B\geq0$ requires $b_{SC_1}=0$ and $b_{C_1S}=0$, while $\Delta_{\pm}\geq0$ implies $a_{C_{1}S} = 0$.
	Then, the inequality
	\begin{align}
		U\left(B \mp \frac{1}{2}(\iden - AA^{\dagger})\right)U^{\dagger} = \begin{pNiceArray}[cell-space-limits = 3pt]{c:c:c}
			b_{SS} \pm \frac{1}{2}(a_{SS}a_{SS}^{\dagger} + a_{SC_{2}}a_{SC_{2}}^{\dagger} - \iden) & \pm\frac{1}{2}(a_{SC_{1}}a_{C_{1}C_{1}}^{\dagger}) & \Delta_{SC_{2}}\\
			\hdottedline
			\pm\frac{1}{2}(a_{C_{1}C_{1}}a_{SC_{1}}^{\dagger})  & 0 &  \Delta_{C_{1}C_{2}}\\
			\hdottedline
			\Delta_{C_{2}S} & \Delta_{C_{2}C_{1}} & \Delta_{C_{2}C_{2}}
		\end{pNiceArray} \geq 0,
	\end{align}	
	forces $a_{SC_{1}}a_{C_{1}C_{1}}^{\dagger} = 0$. 
	Since $a_{C_{1}C_{1}}$ is a unitary matrix, we obtain $a_{SC_{1}} =0$.
	As a result, $a_{SC}^{*}\chi_{C} a_{SC}^{\dagger} = 0$, and $\chi_{S}'  =  a^{*}_{SS}\chi_{S}a_{SS}^{\dagger}$.
	This directly implies that $\Sl_{\pm}(\mu'_{S},\chi'_{S})\leq \Sl_{\pm}(\mu_{S},\chi_{S})$.
	
	\textbf{Case 3-2}:
	Assume that $a_{CS}(\mu_S \pm \frac{\iden}{2})a_{CS}^\dagger\neq0$ but $\psi_{CS}(\mu_{S}\pm\frac{\iden}{2})\psi_{CS}^{\dagger}=0$.
   This is true only if the number of catalytic modes $m_C\geq2$.
	Eqs.~\eqref{eq:cat-condition-1-alt} and~\eqref{eq:cat-condition-2-alt} can be rewritten as 
	\begin{align}
		\mu_{C}\pm\frac{\iden}{2} & =  \sum_{j=0}^{N} (a_{CC})^{j} [a_{CS}(\mu_{S} \pm \frac{\iden}{2})a_{CS}^{\dagger}+\Delta_{\pm} ](a_{CC}^{\dagger})^{j} + (a_{CC})^{N+1} (\mu_{C}\pm\frac{\iden}{2})(a_{CC}^{\dagger})^{N+1},\\
		\chi_{C} & =  \sum_{j=0}^{N} (a_{CC}^{*})^{j} a^{*}_{CS} \chi_{S} a_{CS}^{\dagger} (a_{CC}^{\dagger})^{j} + (a_{CC}^{*})^{N+1}\chi_{C}(a_{CC}^{\dagger})^{N+1},
	\end{align}
	for any $N$ by repeatedly replacing $\mu_{C}$ and $\chi_{C}$ in the RHS expression of the equations.
	If there exists $j>0$ such that either $\psi^{*}(a_{CC}^{*})^{j} a^{*}_{CS} \chi_{S} a_{CS}^{\dagger} (a_{CC}^{\dagger})^{j}\psi^{\dagger}\neq0$ or $\psi(a_{CC})^{j} \Delta_{\pm} (a_{CC}^{\dagger})^{j}\psi^{\dagger}\neq0$, we prove the desired inequality by the argument used for \textbf{Case 2}.
	Otherwise, we have 
	\begin{align}
		\psi(\mu_{C}\pm\frac{\iden}{2})\psi^{\dagger} &= \psi (a_{CC})^{N}(\mu_{C}\pm\frac{\iden}{2})(a_{CC}^{\dagger})^{N} \psi^{\dagger},\label{eq: psi aCC N 1}\\
		\psi^{*}\chi_{C}\psi^{\dagger} &= \psi^{*} (a_{CC}^{*})^{N}\chi_{C}(a_{CC}^{\dagger})^{N} \psi^{\dagger},\label{eq: psi aCC N 2}
	\end{align}
	for any $N$.
	This implies that a vector space $\mathcal{V} \coloneqq \mathrm{span}\{(a_{CC}^{\dagger})^{N} \psi^{\dagger} | N\geq0\}$ is an invariant subspace under the application of $a_{CC}^{\dagger}$, and the matrix $a_{CC}$ and the column vector $\psi^{\dagger}$ have the block form
	\begin{align}
		a_{CC} = \begin{pNiceArray}{c:c}
			a_{VV} & 0\\
			\hdottedline
			a_{WV} & a_{WW}
		\end{pNiceArray}, \quad \psi^{\dagger} = \begin{pNiceArray}{c}
			\psi_{V}^{\dagger} \\ \hdottedline 0
		\end{pNiceArray},
	\end{align}
	where $a_{VV}$ maps $\mathcal{V}$ into itself. 
	Finally, we assume, without loss of generality, that $\mu_{S}\pm \frac{\iden}{2}>0$ using Lemma~\ref{le:vac-block-diag}.
	Then, the condition $\psi(a_{CC})^{N}a_{CS}(\mu_{S}\pm\frac{\iden}{2})a_{CS}^{\dagger}(a_{CC}^{\dagger})^{N}\psi^{\dagger} = 0$ for any $N\geq0$ implies that $a_{CS}(\mu_{S}\pm\frac{\iden}{2})a_{CS}^{\dagger}$ has no support in $\mathcal{V}$.
	In other words, $a_{CS}$ is in the block form $a_{CS}^{\dagger} = \begin{pNiceArray}{c:c} 0 & a_{WS}^{\dagger} \end{pNiceArray}$.
	
	Now we rewrite the block form of $A$
	\begin{align}
		A = \begin{pNiceArray}{c:c:c}
			a_{SS} & a_{SV} & a_{SW}\\
			\hdottedline
			0 & a_{VV} & 0\\
			\hdottedline
			a_{WS} & a_{WV} & a_{WW}
		\end{pNiceArray}. 
	\end{align}
	Recall the catalyst recovery conditions Eqs.~\eqref{eq:cat-condition-1} and~\eqref{eq:cat-condition-2}. 
	If we restrict our attention to $\mathcal{V}$ subspace, the following equalities hold
	\begin{align}
		\mu_{VV} &= a_{VV}\mu_{VV}a_{VV}^{\dagger} + b_{VV},\\
		\chi_{VV} &= a_{VV}^{*}\chi_{VV}a_{VV}^{\dagger}.
	\end{align} 
	These equations imply that $(\mu_{VV},\chi_{VV})$ is the fixed-point of a GCO $(a_{VV},b_{VV})$.
	Applying Lemma~\ref{le: fixed point decomposition}, another block decomposition
	\begin{align}\label{eq: V1 is separated}
		A = \begin{pNiceArray}{c:c:c:c}
			a_{SS} & 0 & a_{SV_{2}} & a_{SW}\\
			\hdottedline
			0 & a_{V_{1}V_{1}} & 0 & 0\\
			\hdottedline
			0 & 0 & a_{V_{2}V_{2}} & 0\\
			\hdottedline
			a_{WS} & 0 & a_{WV_{2}} & a_{WW}
		\end{pNiceArray},\quad 
		B = \begin{pNiceArray}{c:c:c:c}
			b_{SS} & 0 & b_{SV_{2}} & b_{SW}\\
			\hdottedline
			0 & 0 & 0 & 0\\
			\hdottedline
			b_{V_{2}S} & 0 & b_{V_{2}V_{2}} & b_{V_{2}W}\\
			\hdottedline
			b_{WS} & 0 & b_{WV_{2}} & b_{WW}
		\end{pNiceArray},
	\end{align}
	are derived. 
    Here $V_1$ and $V_2$ are divided such that $\chi_{VV}=\chi_{V_1V_1}\oplus 0_{V_2V_2}$. Because $\chi_{VV}\neq0$, $\dim(a_{V1V1})\geq1$.
	Note that we already used $B\geq\pm\frac{1}{2}(\iden - AA^{\dagger})$ condition to remove all blocks of $A$ and $B$ containing $V_{1}$ as an index, except for the unitary matrix $a_{V_{1}V_{1}}$. 
	Eq.~\eqref{eq: V1 is separated} indicates that modes in $V_{1}$ block evolve unitarily and independently from all other modes.
	Hence, there exists a catalyst $(\mu'_{C'},\chi'_{C'})$ and a GCO $(A',B')$ that can achieve the same state transformations and they are given by 
	\begin{align}
		A' = \begin{pNiceArray}{c:c:c}
			a_{SS}  & a_{SV_{2}} & a_{SW}\\
			\hdottedline
			0  & a_{V_{2}V_{2}} & 0\\
			\hdottedline
			a_{WS}  & a_{WV_{2}} & a_{WW}
		\end{pNiceArray},\quad 
		B' = \begin{pNiceArray}{c:c:c}
			b_{SS} & b_{SV_{2}} & b_{SW}\\
			\hdottedline
			b_{V_{2}S} & b_{V_{2}V_{2}} & b_{V_{2}W}\\
			\hdottedline
			b_{WS} & b_{WV_{2}} & b_{WW}
		\end{pNiceArray},\quad
		\mu'_{C'} = \begin{pNiceArray}{c:c}
			\mu_{V_{2}V_{2}} & \mu_{V_{2}W}\\
			\hdottedline
			\mu_{WV_{2}} & \mu_{WW}
		\end{pNiceArray},\quad 
		\chi'_{C'} = \begin{pNiceArray}{c:c}
			0 & \chi_{V_{2}W} \\ 
			\hdottedline
			\chi_{WV_{2}} & \chi_{WW}
		\end{pNiceArray}.
	\end{align}
	Therefore we can repeat the same argument iteratively until i) \textbf{Case 3-1} can be applied and the proof is concluded or ii) catalyst state is reduced to a single mode and either \textbf{Case 2} or \textbf{Case 3-1} is true. 
\end{proof}

\end{document}